\documentclass[11pt]{article}
\usepackage{authblk}
\usepackage[margin=1in]{geometry}

\usepackage{graphicx,color,amsmath,amsfonts,enumerate,amsthm,amssymb,mathtools,enumitem,thmtools,hyperref,mathdots,centernot,bm,soul,bbm,stmaryrd}
\usepackage[capitalise, noabbrev]{cleveref}
\usepackage{placeins}
\usepackage{floatrow}
\usepackage{array}
\usepackage[caption=false]{subfig}

\usepackage{multirow}

\usepackage{tikz,pgfplots}
\usetikzlibrary{quantikz2}

\pgfplotsset{compat=1.15}
\hypersetup{colorlinks=true,linkcolor=blue,citecolor=blue,urlcolor=blue}

\usepackage{algorithm,algpseudocode}
\usepackage{algorithmicx}

\theoremstyle{plain}
\newtheorem{theorem}{Theorem}

\newtheorem{lemma}[theorem]{Lemma}

\newtheorem{problem}[theorem]{Problem}

\theoremstyle{definition}
\newtheorem{definition}[theorem]{Definition}

\usepackage{complexity}

\newcommand{\bbn}{\mathbb{N}}
\newcommand{\bbr}{\mathbb{R}}
\newcommand{\bbz}{\mathbb{Z}}

\newcommand{\commentblock}[1]{}
\newcommand{\abs}[1]{\left| {#1} \right|}

\newcommand{\tr}[1]{\mathrm{tr}\left( #1 \right)}

\newcommand{\diag}[1]{\mathrm{diag}\left( #1 \right)}

\newcommand{\set}[1]{\left\{#1\right\}}
\newcommand{\bitstring}[1]{\left\{0,1\right\}^{#1}}
\newcommand{\order}[1]{\mathcal O\left(#1\right)}

\newcommand{\orderpoly}[1]{O^{*}\left(#1\right)}

\newcommand{\twonorm}[1]{||#1||_{2}}

\newcommand{\ketbra}[2]{\ensuremath{\left|#1\right\rangle\!\!\left\langle#2\right|}}

\newcommand{\clif}[1]{\mathcal{C}_{#1}}
\newcommand{\stab}[1]{\mathcal{S}_{#1}}

\newcommand{\phasegate}{\sqrt{Z}}

\definecolor{hpCol}{RGB}{150, 0, 150}

\definecolor{darkcyan}{RGB}{0,139,139}

\newcommand{\ml}[2]{\twonorm{#1 \ket{#2}}^2}
\DeclareMathOperator{\rank}{\mathrm{rank}}
\newcommand{\sall}{\mathcal{S}_{\mathrm{all}}}
\newcommand{\sscase}{S}
\newcommand{\pset}[1]{\mathcal{S}_{(#1)}}
\newcommand{\projset}[2]{\Pi_{({#1}),{#2}}}
\newcommand{\reps}{\mathrm{Reps}}

\newcommand{\probnum}[1]{\textsc{Opt-NSPS}(#1)}
\newcommand{\probdec}[1]{\textsc{NSPS}(#1)}
\newcommand{\probproj}[1]{\textsc{Srch-NSPS}(#1)}
\newcommand{\reducesto}{\leq^{\poly}_{\mathrm{T}}}

\newcommand{\orb}[1]{\mathrm{Orb}(#1)}
\newcommand{\eqloc}{\sim_{loc}}
\newcommand{\notset}[1]{\neg ~#1}

\newcommand{\pauliX}{X}
\newcommand{\pauliY}{Y}
\newcommand{\pauliZ}{Z}

\newcommand{\maxlin}{\textsc{Max-E3-Lin}_2}
\crefname{appendix}{Appendix}{Appendices}
\Crefname{appendix}{Appendix}{Appendices}

\DeclareMathOperator{\occ}{\mathrm{occ}}

\begin{document}

\title{Complexity and Applications of Nearest Stabilizer Product State Problems}
\date{\today}

\author{Daniel Grier}
\affil{Department of Mathematics and Department of Computer Science and Engineering, University of California, San Diego}

\author{Hakop Pashayan}
\affil{Hon Hai (Foxconn) Research Institute, Taipei 114699, Taiwan}

\author{Luke Schaeffer}
\affil{Institute for Quantum Computing, University of Waterloo}
\maketitle

\begin{abstract}

	Consider the following optimization problem over stabilizer product states: given an $n$-qubit stabilizer state $\left|\psi\right\rangle$ and a set of single-qubit stabilizer states $S$, maximize $\left|\langle \psi | \phi_1, \ldots, \phi_n\rangle\right|^2$ over single-qubit stabilizer states $\left|\phi_i\right\rangle \in S$. By varying the set $S$, we show that solutions to this problem can be useful in a variety of settings: tighter runtime bounds for certain classical simulation algorithms; measures of entanglement; and the complexity of low-rank matrix completion. Moreover, we give a complete complexity classification of this nearest stabilizer product state problem. After accounting for the symmetries in the Clifford group, there are $9$ distinct possible sets $S$, and we show that all but the two simplest of these are $\NP$-complete.

\end{abstract}

\section{Introduction}\label{sec:intro}

Stabilizer states are ubiquitous in quantum computing because they exhibit many important quantum properties such as superposition, entanglement and contextuality whilst also admitting compact classical representations that are computationally easy to manipulate. The Gottesman--Knill theorem \cite{gottesman1998heisenberg, aaronson2004improved} provides an efficient algorithm for representing stabilizer states, updating these under the action of Clifford unitaries and computing the modulus of inner products between pairs of $n$-qubit stabilizer states. One can also efficiently compute partial traces of stabilizer states, representing and updating these potentially mixed states under the action of Clifford unitaries. However, despite the tractability of many computations involving stabilizer states, they are incredibly expressive. This allows them to underpin large classes of modern quantum error-correcting codes, quantum computation models such as measurement-based quantum computation, the study of contextuality and entanglement.
This work studies a class of problems that originally drew our attention due to their connection to modern techniques for classically simulating quantum circuits. However, these problems also revealed connections to the study of multipartite entanglement and mathematical literature on rank minimization problems.

\subsection{Problem Statement and Summary of Results}
In this work, we ask the question, ``How close is a given stabilizer state to a set of stabilizer product states?''. Given an $n$-qubit state $\ket{\psi}$ and a set of projectors $\Pi$, we consider the quantity:
\begin{align}
	m(\psi, \Pi):= \max_{\pi \in \Pi} |\langle \psi | \pi |\psi \rangle| = \underset{\pi \in \Pi}{\max} ~\ml{\pi}{\psi}.
\end{align}
In this work, we consider a special case of this problem focused on stabilizer states. In particular, we focus on $n$-qubit stabilizer states $\ket\psi\in \stab{n}$ and sets $\Pi$ that are a subset of projections onto a tensor product of single-qubit stabilizer states. To this end, we define
\begin{align*}
	\sall : & = \left\{ \frac{I \pm X}{2}, \frac{I \pm Y}{2}, \frac{I \pm Z}{2} \right\}                      \\
	        & =\set{\ketbra{+}{+},\ketbra{-}{-},\ketbra{+i}{+i},\ketbra{-i}{-i}, \ketbra{0}{0},\ketbra{1}{1}}
\end{align*}
to be the set of single-qubit stabilizer projectors. For our problem, this means that $\Pi=\sscase^{\otimes n}$ for some $\sscase \subseteq \sall$ or more generally, $\Pi=\Pi_{\sscase,k}:=\sscase^{\otimes k} \otimes I^{\otimes n-k}$ for some $k\in [n]$.

\begin{definition}[Nearest Stabilizer Product State (NSPS)]

	\label{defn:probnum}
	For $\sscase\subseteq \sall$, $\sscase\neq \emptyset$,
	we define the optimization problem $\probnum{\stab{n},\sscase}$ as follows:

	Given any state $\ket{\psi}\in \stab{n}$ and $k\in [n]$, compute the quantity:
	\begin{align}\label{eq:defin_m}
		m(\psi,\Pi_{\sscase,k}):=\underset{\pi \in \Pi_{\sscase,k}}{\max} \|\pi \ket{\psi} \|_2^2.
	\end{align}
\end{definition}

Based on well-known properties of stabilizer states, it is easy to show that the quantity $m(\psi, \Pi_{\sscase,k})$ is either zero or $2^{-p}$ for some $0 \leq p \leq k$. Hence, approximations $\hat{m}$ to precision $\frac{1}{\sqrt{2}}m < \hat{m} < \sqrt{2}m$ can be used to recover the exact value.
We use $\probdec{\stab{n}, \sscase}$ to denote the associated decision problem\footnote{Where, in addition to the inputs $\ket{\psi}$ and $k$ we are given $p \in \bbn$ and our task is to decide if $\textsc{Yes}: m(\psi,\Pi_{\sscase,k})\geq 2^{-p+1}$ or if $\textsc{No}: m(\psi,\Pi_{\sscase,k})\leq 2^{-p}$.}. Additionally, we use $\probproj{\stab{n}, \sscase}$ to denote the associated search problem where the task is to return any projector $\pi \in \Pi_{\sscase,k}$ that achieves the optimal value  $m(\psi,\Pi_{\sscase,k})$.  In \cref{app:FormalDefin}, we show that the search, optimization and decision variants of the problems are polynomial-time equivalent in complexity.

For all non-empty subsets $\sscase \subseteq \sall$, $\probdec{\stab{n},\sscase}$ is in $\NP$. To see this observe that in the case of a $\textsc{Yes}$ instance (i.e. $m(\psi,\Pi_{\sscase,k})\geq 2^{-p+1}$), there exists a short witness $\pi\in \Pi_{\sscase,k}$ such that $\|\pi \ket{\psi} \|_2^2\geq 2^{-p+1}$. Further, given $\pi$, a deterministic Turing machine can verify in polynomial time that indeed $\|\pi \ket{\psi} \|_2^2\geq 2^{-p+1}$ \cite{gottesman1998heisenberg, aaronson2004improved}.

There are $63$ non-empty subsets $\sscase \subseteq \sall$ that give rise to different problems $\probdec{\stab{n}, \sscase}$. We use local Clifford symmetries to classify these into $9$ representatives labeled by $S=\pset{abc}$ for some $a\geq b \geq c\in \set{0,1,2}$ (see \cref{eq:S_representatives} for a list of these). This notation indicates that of the single qubit Paulis $\set{X,Y,Z}$, $\sscase$ is constructed by including exactly $a$ eigenprojectors of one, $b$  of another and $c$ of the remaining Pauli. So $\pset{abc}$ has cardinality $a+b+c$ and $\pset{222}=\sall$. We characterize the complexity of these problems and summarize these results in \cref{tbl:complexity_results}. Henceforth, we use the more concise notation $\probdec{abc}\equiv \probdec{\stab{n},\pset{abc}}$, $\Pi_{(abc),k}\equiv \Pi_{\pset{abc},k}$ and $m(\psi,\sscase) \equiv m(\psi,\Pi_{\sscase,n})$.

\begin{table}[H]
	\centering
	\caption{The decision problem $\probdec{\stab{n},\sscase}$ is classified into nine classes. Complexity results associated with each class are summarized.}\label{tbl:complexity_results}
	\setlength{\extrarowheight}{3pt}
	\begin{tabular}{|c|c|l|c|}
		\hline
		\textbf{Decision problem class} & \textbf{Complexity} & \textbf{Comment}                             & \textbf{Reference}                 \\ \hline
		$\probdec{{100}}$               & \cP                 & by Gottesman--Knill theorem                  & --                                 \\ \hline
		$\probdec{{200}}$               & \cP                 & greedy algorithm                             & \cref{lem:200complexity}           \\ \hline
		$\probdec{{110}}$               & \NP-complete        & by reduction from $\maxlin$                  & \cref{thm:partial-bases-hardness}  \\ \hline
		$\probdec{{210}}$               & \NP-complete        & by reduction from $\maxlin$                  & \cref{thm:partial-bases-hardness}  \\ \hline
		$\probdec{{111}}$               & \NP-complete        & by reduction from $\maxlin$                  & \cref{thm:partial-bases-hardness}  \\ \hline
		$\probdec{{220}}$               & \NP-complete        & by reduction from $\textsc{Independent-Set}$ & \cref{thm:hardness_of_220_and_222} \\ \hline
		$\probdec{{211}}$               & \NP-complete        & by reduction from $\maxlin$                  & \cref{thm:partial-bases-hardness}  \\ \hline
		$\probdec{{221}}$               & \NP-complete        & by reduction from $\maxlin$                  & \cref{thm:partial-bases-hardness}  \\ \hline
		$\probdec{{222}}$               & \NP-complete        & by reduction from $\textsc{Independent-Set}$ & \cref{thm:hardness_of_220_and_222} \\ \hline
	\end{tabular}
\end{table}

Our $\NP$-hardness proofs for $\probdec{220}$ and $\probdec{222}$ also prove a constructive result that enables solving particular instances of $\probdec{220}$ and $\probdec{222}$ in mildly growing exponential time. To prove the NP-hardness of $\probdec{220}$ and $\probdec{222}$, we use a specific graph construction that maps an input graph $G$ into a new graph $H=H(G)$ (also known as the 3-subdivision of $G$ \cite{ghanbari2023securedominationnumberksubdivision}) where each edge in $G=(V,E)$ is divided into 3 edges by placing two new vertices on the edge. Using this construction, we show that for any graph $G=(V,E)$, the nearest stabilizer product states to the graph state $\ket{H(G)}$ can be computed directly from the vertex cover number $\tau(G)$ of $G$:
\begin{align}\label{eq:222toMIS}
	-\log_2 m(H(G),\pset{220})=-\log_2 m(H(G),\pset{222})=\tau(G)+\abs{E}
\end{align}
Consequently, for any graph $G$, the computation of $m(H(G),\pset{220})$ and $m(H(G),\pset{222})$ is possible in polynomial space and mildly exponential runtime. For example, for a graph $G$ with $v=\abs{V}$ vertices and $e=\abs{E}$ edges, $H(G)$ has $v+2e$ vertices so a naive brute-force approach to computing $m(H(G),\pset{222})$ would require $\orderpoly{3^{v+2e}}$ runtime where $\orderpoly{\cdot}$ suppresses polynomial factors. In contrast, by using \cref{eq:222toMIS} and employing existing algorithms for computing $\tau(G)$ \cite{Xiao_2017} we find an $\orderpoly{1.1996^{v}}$ algorithm for computing the same quantity. For $\probdec{222}$, this method can be extended beyond graphs in the image of $H$ by using the fact that $m(H,\pset{222})$ is invariant under local complementation of $H$.

\subsection{Motivation and Relation to Other Work}

Several of these NSPS problems pose natural questions and/or connect to other areas of research.

\emph{Additive error Born rule probability estimation and $\probdec{110}$:}
A prominent class of modern classical simulation algorithms \cite{bravyi2016improved, bravyi2019simulation, seddon2020quantifying, pashayan2022fast} use stabilizer state decompositions of so-called magic states to classically simulate quantum circuits, $U$, dominated by Clifford gates and a limited number of non-Clifford $T$-gates ($T\equiv \mathrm{diag}(1, e^{i\pi/4})$). The only source of super-polynomial scaling in the runtime of these simulators comes from a linear dependence on the so-called \emph{stabilizer extent}, a measure of non-stabilizerness that depends on $t$, the number of $T$-gates. One such simulator is the in-class state of the art for estimating (to $\pm \epsilon$ precision) the Born rule probabilities
associated with measuring $U \ket{0^n}$ in the computational basis  \cite{pashayan2022fast}.
In \cref{app:additive_CSim}, we briefly describe how $\probnum{110}$ can be used to improve the bounds on the runtime of this algorithm.\footnote{The details will be presented in an upcoming work, where we also present efficient methods for upper bounding $\probnum{110}$.} In essence, the task of Born rule probability estimation can be reduced to the estimation of an inner product between a $t$-qubit, circuit-dependent, stabilizer state $\ket{\psi}$ and the $t$-fold tensor product of so-called $T$-states ($\ket{T}\propto \ket{0}+e^{i\pi/4}\ket{1}$), which we denote by $\ket{T^t}$. By decomposing each $T$-state into a linear combination of $\ket{+}$ and $\ket{+i}$,
one can estimate the inner product $\braket{T^t}{\psi}$ by using Monte-Carlo sampling of stabilizer product states $\ket{\phi}\in \set{\ket{+},\ket{+i}}^{\otimes t}$. The convergence rate of the estimators, and hence the algorithm's runtime, depends on the largest overlap $\abs{\braket{\phi}{\psi}}$ which is determined by $\probnum{110}$.
Finally, we note that decomposing each $T$-state into a linear combination of $\ket{+}$ and $\ket{+i}$ is not arbitrary but optimal with respect to an operationally relevant metric, achieving the stabilizer extent.

\emph{From $\probdec{200}$, $\probdec{220}$ and $\probdec{222}$ to rank minimization:}
$\probnum{200}$ asks,
``If an input stabilizer state is measured in the computational basis, what is the probability of the most likely outcome?''
The expert reader may be aware that this is determined by the rank of the $X$-check matrix of the input stabilizer state (see \cref{lem:200_UB}). In particular, the probability is maximized when the rank is minimized.
This is the starting point for our results on $\probdec{220}$ and $\probdec{222}$. $\probproj{220}$ addresses the question ``How do we apply Hadamard gates to a subset of the qubits so as to maximize the probability of the most likely outcome?''
Since Clifford operations like Hadamard modify the $X$-check matrix, the goal is to find those operations which minimize its rank. As it turns out, this setup falls under the umbrella of rank minimization problems. There, the input is a matrix whose entries depend on polynomially many variables, and the goal is to find an assignment to the variables that minimizes the rank. As we discuss in \cref{app:LR_Matrix_Completion}, our results show that rank minimization over $\mathbb F_2$ remains $\NP$-hard for classes of instances that, to our knowledge, do not yet appear in the literature---namely, those instances with exactly one variable per column, which may appear in several entries of that column but in no other column.

\emph{$\probdec{222}$ and connections to the geometric and Schmidt measures of entanglement:}\\
$\probnum{222}$ is concerned with finding the distance (fidelity) to the nearest stabilizer product state. This quantity is related to measures of entanglement on graph states $\ket{G}$ such as the geometric measure $E_g(\ket{G})$ \cite{wei2003geometric, chen2014comparison} and the Schmidt measure $E_s(\ket{G})$ \cite{eisert2001schmidt}. The geometric measure quantifies the distance to the nearest product state and is defined as $E_g(\ket{\psi})=-\log_2 \max_{\ket{\phi}} \abs{\braket{\phi}{\psi}}^2$ where the maximization is taken over all product states $\ket{\phi}$. The canonical polyadic (CP) rank or tensor rank, $R_{\min}$, of an $n$-qubit state $\ket{\psi}$ quantifies the minimum number of terms $R_{\min}$ needed to decompose $\ket{\psi}$ into a linear combination of product states. The Schmidt measure is defined as $E_s(\ket{\psi}):=\log_2 R_{\min}$ \cite{eisert2001schmidt}. The quantities $E_{sg}(\ket{G})$ and $E_{ss}(\ket{G})$ associated with this work can be defined as the \emph{stabilizer} analogues of $E_{g}(\ket{G})$ and $E_{s}(\ket{G})$ respectively. That is, they are optimized over all stabilizer product states as opposed to all product states. Thus, for example $E_{sg}(\ket{\psi}):=-\log_2 m(\psi, \pset{222})$.
We relate these to previously derived results \cite{markham2007entanglement} that upper bound the geometric and Schmidt measure of entanglement on graph states $\ket{G}$ by the vertex cover number $\tau(G)$ of the graph:
\begin{align*}
	\tau(G)\geq & E_{sg}(\ket{G}) \geq E_g(\ket{G})                       \\
	            & E_{sg}(\ket{G}) \geq E_{ss}(\ket{G}) \geq E_s(\ket{G}).
\end{align*}

We prove $E_{sg}(\ket{G}) \geq E_{ss}(\ket{G})$ by showing that a solution to a $\ket{\psi}$ instance of $\probproj{222}$ immediately provides an orthogonal stabilizer product state decomposition of $\ket{\psi}$ with exactly $1/m({\psi}, \pset{222})$ terms. By using \cref{eq:222toMIS} and employing algorithmic tools from Ref.~\cite{Xiao_2017} our results provide a computational tool for bounding $E_g (\ket{\psi})$ and $E_{ss} (\ket{\psi})$ for a restricted set of stabilizer states.
More details can be found in \cref{app:Entanglement_and_mult_CSim}.

\section{Technical Background}

We briefly summarize some background needed for our discussion.

\emph{Graphs:} A graph $G=(V,E)$ with vertex set $V$ and edge set $E$ admits an independence number $\alpha(G)$, a vertex cover number $\tau(G)$ and matching number $\nu(G)$.
An \emph{independent set} of $G$ is a subset of vertices where no two vertices are connected by an edge in $E$. The \emph{maximum independent set} is an independent set with the largest possible number of vertices. The \emph{independence number} is the cardinality of the maximum independent set.
A \emph{vertex cover} of $G$ is a subset of vertices such that all edges in $E$ have at least one incident vertex in the subset. The \emph{minimum vertex cover} is the smallest subset of vertices that is a vertex cover. The \emph{vertex cover number} is its cardinality.
A \emph{matching} is a subset of edges where no two edges share a common endpoint (i.e., no two edges are incident to the same vertex). A \emph{maximum matching} is a matching such that there is no other matching with a larger number of edges. The total number of edges in a maximum matching is called the \emph{matching number} of the graph.
A well-known property of graphs states that the complement (in $V$) of the maximum independent set of $G=(V,E)$ is a minimum vertex cover of $G$ (sometimes referred to as Gallai's lemma). Hence, $\alpha(G)+\tau(G)=\abs{V}$.
The \emph{local complement} (LC) is a map from a graph $G=(V,E)$ and a vertex $v \in V$ to a graph $G^* v=G'=(V,E')$ that modifies $G$ by toggling the edges with both endpoints in the neighborhood $N(v)$ of $v$, i.e., if at least one endpoint of $e\in E$ is not in $N(v)$ then $e\in E'$ but for all distinct pairs $i,j \in N(v)$, $(i,j)\in E'$ if and only if $(i,j)\notin E$.

\emph{Graph States:} Graph states are a subset of stabilizer states that can be constructed as follows. Given a graph $G=(V,E)$, we prepare the graph state $\ket{G}$ by first preparing $\ket{+}^{\otimes \abs{V}}$ then for each $(i,j)\in E$, we act with a controlled-Z gate between qubits $i$ and $j$. A natural choice of stabilizer generators for $\ket{G}$ is one where the $j^{\mathrm{th}}$ generator $g_j$ is associated with the $j^{\mathrm{th}}$ vertex of $G$ such that the $j^{\mathrm{th}}$ Pauli of $g_j$ is a Pauli-$X$ operator and, if $(i,j)\in E$ then the $i^{\mathrm{th}}$ Pauli of $g_j$ is a Pauli-$Z$ operator and otherwise it is the identity operator. This results in an $X$-check matrix that is the $n \times n$ identity matrix and a $Z$-check matrix that is the $n \times n$ adjacency matrix $\mathrm{adj}(G)$ of $G$.
Every stabilizer state $\ket{\psi}$ is, up to multiplication by single qubit Clifford unitaries, equivalent to some graph state $\ket{G}$.
Given a graph state $\ket{G}$ and a vertex $v\in V(G)$, one can prepare the graph state $\ket{G^*v}$ associated with the LC of $G$ with respect to vertex $v$ by acting on $\ket{G}$ with the unitary $U_v:=\sqrt{-i X_v}\otimes_{u \in N(v)} \sqrt{i Z_{u}}$.
Crucially, $U_v$ is a local Clifford operation hence, all graph states that are associated by some sequence of LC transformations share the same entanglement properties and, the local Clifford orbit of a given stabilizer state may intersect multiple graph states.

\emph{Notation:} We use ${\phasegate}$, $H$ to denote the single qubit phase and Hadamard Clifford unitaries respectively.

\section{Classification}
\label{sec:problem_classification}

We observe the following symmetry of the objective in \cref{eq:defin_m}. For any Clifford unitary $U$:
\begin{align}\label{eq:symmetry_optimization_fn}
	m(\psi,\Pi_{\sscase,k})=m(\psi', U\Pi_{\sscase,k} U^{\dagger}),
\end{align}
where $\ket{\psi'}:=U\ket{\psi}$. We note that the complexity of $\probdec{\stab{n},\sscase}$ is qubit order independent by choosing $U$ from the set of Cliffords generated by swap gates.

For tensor-power Cliffords of the form $U= U_{loc}^{\otimes n}$ where $U_{loc}\in \clif{1}$, this symmetry imposes an equivalence between any pair of problems $\probdec{\stab{n},\sscase}$ and $\probdec{\stab{n},\sscase'}$ when $\sscase$ and $\sscase'$ are related by conjugation by $U_{loc}$.
Below, we define nine specific choices of the subset $\sscase$ which classify all
choices of non-empty $\sscase$ in the sense that every $\sscase\neq \emptyset$
is contained in the orbit of exactly one of the below representatives and any pair of subsets $\sscase, \sscase'$ associated with the same representative have the same complexity:

\begin{equation}\label{eq:S_representatives}
	\begin{aligned}
		\pset{100} & :=\{\ketbra{0}{0}\}                                                                               \\
		\pset{200} & :=\{\ketbra{0}{0},\ketbra{1}{1}\}                                                                 \\
		\pset{110} & :=\{\ketbra{0}{0},\ketbra{+}{+}\}                                                                 \\
		\pset{210} & :=\{\ketbra{0}{0},\ketbra{+}{+}, \ketbra{-}{-}\}                                                  \\
		\pset{111} & :=\{\ketbra{0}{0},\ketbra{+}{+}, \ketbra{+i}{+i}\}                                                \\
		\pset{220} & :=\{\ketbra{0}{0},\ketbra{1}{1},\ketbra{+}{+}, \ketbra{-}{-}\}                                    \\
		\pset{211} & :=\{\ketbra{0}{0},\ketbra{+}{+}, \ketbra{-}{-}, \ketbra{+i}{+i}\}                                 \\
		\pset{221} & :=\{\ketbra{0}{0},\ketbra{+}{+}, \ketbra{-}{-}, \ketbra{+i}{+i}, \ketbra{-i}{-i}\}                \\
		\pset{222} & :=\{\ketbra{0}{0},\ketbra{1}{1},\ketbra{+}{+}, \ketbra{-}{-}, \ketbra{+i}{+i}, \ketbra{-i}{-i}\}.
	\end{aligned}
\end{equation}
We use $\reps:=\set{\pset{100}, \pset{200}, \pset{110}, \pset{210}, \pset{111}, \pset{220}, \pset{211}, \pset{221}, \pset{222}}$ to denote this set of representatives.
We formalize and prove these statements in \cref{app:Classification}.

\section{Complexity Results}

We present our complexity results below. In \cref{sec:ComplexityEasy}, we discuss $\probdec{100}$, $\probdec{200}$ and generalizations.  In \cref{sec:ComplexityIncompleteBases}, we present hardness results for the remaining representatives that contain an incomplete basis, i.e., $\probdec{110}$, $\probdec{111}$, $\probdec{210}$, $\probdec{211}$ and $\probdec{221}$. Finally, in \cref{sec:ComplexityCompleteBases}, we present hardness results for the remaining representatives that contain only complete bases; these are $\probdec{220}$ and $\probdec{222}$.

\subsection{Easy consequences of the Gottesman--Knill theorem}\label{sec:ComplexityEasy}

\subsubsection{\texorpdfstring{$\probdec{100}$}{NSPS(100)}}
\label{sec:100}
We note that $\probdec{100}$ is a trivial setting where there are no choices to optimize over and hence this class can be decided in $\cP$ by application of the Gottesman--Knill theorem.
For the more general input state setting, the situation is different.
For example, given a description of a state $\ket{\psi}$ generated by acting on $\ket{0}^{\otimes n}$ with $n^{\order{1}}$ gates from a universal gate set, computing
$\abs{\braket{0^n}{\psi}}^2$
to multiplicative precision is $\#\P$-hard \cite{Aaronson2005postselection, Bremner2010IQP, kuperberg2014hardapproximatejonespolynomial, Dawson2005olynomial_equations_over_Z2}.

\subsubsection{\texorpdfstring{$\probdec{200}$}{NSPS(200)}}
Let us now consider the set of computational basis projectors $\sscase=\pset{200}$ where $m(\psi,\projset{200}{k})$ is the probability of the most likely measurement outcome when the first $k$ qubits of $\ket{\psi}$ are measured in the computational basis. It is easy to show that for all stabilizer states $\ket{\psi}$, the optimal value $m(\psi,\projset{200}{k})$ and an optimizing projector can be efficiently computed.

\begin{lemma}\label{lem:200complexity}
	For any $n$-qubit stabilizer state $\ket{\psi}$ and $k\in [n]$, a greedy algorithm can efficiently find an $x^{*}\in \bitstring{k}$ satisfying $\ml{\ketbra{x^*}{x^*}\otimes I^{\otimes {n-k}}}{\psi}=m(\psi,\projset{200}{k})$.
\end{lemma}

\begin{proof}

	Let us first show a basic fact.
	Any two $n$-qubit stabilizer states with identical generators up to signs are related by a Pauli $P$. To see this, just use a Clifford $C$ to send the generators of one state to $Z_1, \ldots, Z_{n}$ and it is clear that the other state's generators are sent of the same set with some sign flips. These sign flips can be removed by conjugating the state with a Pauli $P'\in \set{I,X}^{\otimes n}$ where the locations of the Pauli-$X$ operators in $P'$ is determined by the location of the sign flips. Finally, returning back to the original basis by applying $C^{\dagger}$, we observe that $P:=C^{\dagger}P' C$ achieves the required sign flips in this basis.

	To prove the lemma, first note that all stabilizer projectors acting on a stabilizer state produce a stabilizer state up to normalization. When, without loss of generality, the first qubit of a stabilizer state $\ket{\psi}$ is projected onto $\ket{b}$, for $b\in \set{0,1}$, the squared $2$-norm changes by a factor $f_b\in \set{0, 1, 1/2}$. If $f_b=0$ then $f_{b\oplus 1}=1$ and so greedily projecting onto $\ket{b\oplus 1}$ can never be worse. If $f_b=1$ then $f_{b\oplus 1}=0$ and so greedily projecting onto $\ket{b\oplus 1}$ can never be better.

	The key claim to completing the proof concerns the final case when $f_0=1/2=f_{1}$, i.e., the ``non-deterministic measurement'' case. In this case, we will show that $\frac{1}{2}m(\psi_b,\projset{200}{k-1})=m(\psi,\projset{200}{k})$, i.e., the LHS is independent of $b$ where $\ket{\psi_b}:=\sqrt{2}\bra{b}\otimes I \ket{\psi}$. Hence, we are never worse off by making the greedy choice on deterministic measurements and our choices on non-deterministic measurements are inconsequential so a random choice can be made.

	We now claim that in the ``non-deterministic measurement'' case,
	the post-selected states are given by $\ket{\psi_b}= P^b\ket{\phi}$, where $\ket{\phi}$ is some $n-1$ qubit stabilizer state independent of $b$ and $P$ is some $n-1$ qubit Pauli. This immediately leads to the $m(\psi_0,\projset{200}{k-1})=m(\psi_1,\projset{200}{k-1})$ by \cref{eq:symmetry_optimization_fn} and the invariance of $\projset{200}{k-1}$ under conjugation by any Pauli $P$. Hence, for either choice of $b$:
	\begin{align*}
		m(\psi,\projset{200}{k}) & =\frac{1}{2}\max_{a} m(\psi_a,\projset{200}{k-1})=\frac{1}{2}m(\psi_b,\projset{200}{k-1}).
	\end{align*}

	To see that $\ket{\psi_b}= P^b\ket{\phi}$, let us consider the standard Gottesman--Knill measurement procedure \cite{nielsen2000quantum} applied to the original state $\ket{\psi}$. Observe that in the non-deterministic case,
	a stabilizer generator of $\ket{\psi}$ anticommutes with $Z_1$, the local Pauli-$Z$. We reselect a set of generators $\set{g_i}_i\in [n]$ so that $g_1$ anticommutes with $Z_1$ and the rest all commute with it. Now, the description of the post-measurement state is given by replacing $g_1$ with $(-1)^b Z_1$ depending on the measurement outcome $b$. We can now reselect the generators $g_2, \ldots, g_n$ such that they are the identity on the first qubit. This process involves multiplying any $g_j$ (for $j>1$) with Pauli-$Z$ support on the first qubit by $(-1)^b Z_1$ hence the signs of some generators may depend on $b$, but the generators up to signs do not depend on the outcome. Since two stabilizer states with identical generators up to signs are related by a Pauli $P$, the post-measurement state corresponds to $\ket{b} \otimes P^{b}\ket{\phi}$ and discarding the first qubit leads to the claim $\ket{\psi_b}=P^{b}\ket{\phi}$.

\end{proof}

We note, however, that for non-stabilizer $\ket{\psi}$, greedy algorithms do not work in general. For example, we note that for the non-stabilizer state $\ket{\psi}=\sqrt{\frac{2}{7}}\ket{00}+\sqrt{\frac{2}{7}}\ket{01}+\sqrt{\frac{3}{7}}\ket{11}$, a greedy algorithm fails since we have $\ml{\ketbra{0}{0}\otimes I}{\psi}>\ml{\ketbra{1}{1}\otimes I}{\psi}$ but for $k=2$, the maximum is uniquely achieved by $x^*=(1,1)$.

For the remainder of this work, we focus on stabilizer state inputs. These are associated with a set of stabilizer generators that can be written in the form of a so-called stabilizer check matrix \cite{nielsen2000quantum}. As we will next show, for an input state $\ket{\psi}$, the rank of the $X$-component of the stabilizer check matrix determines the quantity $m(\psi,\pset{200})$. This fact will be important for other choices of subset $\sscase$.

\begin{definition}
	Let $G=(g(1),\ldots,g(n))$ be a set of commuting, independent generators associated with an $n$-qubit stabilizer state $\ket{\psi}$ where $g(i)=(-1)^{a_i} g_1(i)\otimes \ldots \otimes g_n(i) \in \mathcal{P}$. We define the binary valued $n\times n$ matrix $M=M(G)$ entry wise according to:
	\begin{align}
		[M]_{i,j}=\begin{cases} 0 & \text{if}~ [g_j(i),Z_j]=0\\1 & \text{otherwise.} \end{cases}
	\end{align}
	We call $M$ a \emph{commutativity matrix} of $\ket{\psi}$ if $\ket{\psi}$ possesses a generator matrix $G$ such that $M=M(G)$. If $M$ is a commutativity matrix of $\ket{\psi}$ and $M^*$ is the reduced row echelon form of $M$ then we call $M^*$ the \emph{canonical commutativity matrix} of $\ket{\psi}$.
\end{definition}
We note that $M(G)$ does not depend on the phases of the generators $g(i)$ and hence, $M(G)$ is associated with a stabilizer basis. Since the choice of phases $\set{a_1, \ldots, a_n}\in \bitstring{n}$ correspond to different elements of the stabilizer basis, all elements in the basis have a simultaneous commutativity matrix. Different commutativity matrices can be associated with the same stabilizer basis and these are all related to each other by elementary row operations since these are associated with changing the generating set (modulo phases) of the stabilizer group (modulo phases).

\begin{lemma}\label{lem:200_UB}
	For $\ket{\psi}$ an $n$-qubit stabilizer state,
	\begin{align}
		m(\psi,\pset{200})=2^{-r},
	\end{align}
	where $r:=\rank{M(G)}$ is the rank of a commutativity matrix of $\ket{\psi}$.
\end{lemma}

\begin{proof}
	Consider the generating set corresponding to  $M^*$. Let us measure $\ket{\psi}$, using the standard stabilizer state measurement procedure \cite{nielsen2000quantum}. We will measure the qubits in the following order. First we will measure the qubit corresponding to the first leading 1 in $M^*$, then the second leading 1, and so on until we have measured $r$ qubits. For each of these, exactly one distinct stabilizer anticommutes with the measured Pauli-Z operator, each stabilizer is killed and contributes a factor of $1/2$ to the measurement outcome probability. The remaining qubits can now be measured but these all commute with the Pauli-Z measurements so no stabilizer will be killed and these measurements will have deterministic outcomes. So the outcome distribution is deterministic on $n-r$ of the qubits and uniform on the rest and this gives our claim.
\end{proof}

\subsection{Hardness with a partial basis}\label{sec:ComplexityIncompleteBases}

The first of our hardness proofs is for the case where for some basis ($\pauliX$, $\pauliY$, or $\pauliZ$) only one projector is in our set -- a partial basis.
\begin{theorem}
	\label{thm:partial-bases-hardness}
	$\probdec{abc}$ is $\NP$-complete for $abc \in  \set{110, 111, 210, 211, 221}$. That is, whenever $a$, $b$, or $c$ is $1$ except $100$, which is trivial (see \cref{sec:100}).
\end{theorem}
We pointed out earlier that the problems are clearly \emph{in} $\NP$ by supplying the product state as a witness. It only remains to prove $\NP$-hardness by reduction from an $\NP$-complete problem. Here, we use the problem of maximizing the number of satisfied linear equations modulo $2$, for which the special case of size-$3$ clauses ($\maxlin$) is known to be $\NP$-complete by work of H\aa stad \cite{haastad2001some}.
\begin{problem}
Define the $\maxlin$ problem as follows. We are given a collection of variables $x_1, \ldots, x_n$, a threshold $T \in \mathbb N$, and a set of $m$ size-$3$ linear clauses of the form
$$
	x_{a_i} \oplus x_{b_i} \oplus x_{c_i} = d_i
$$
where $a_i, b_i, c_i \in [n]$ and $d_i \in \mathbb{F}_2$. The goal is to decide whether any assignment to the variables satisfies at least $T$ clauses.
\end{problem}
\begin{theorem}[H\aa stad \cite{haastad2001some}]
	$\maxlin$ is $\NP$-complete. Moreover, if there is a polynomial-time $c$-approximation to the $\maxlin$ optimization with constant $c < 2$ then $\P = \NP$.
\end{theorem}

\noindent We can now proceed to the main proof.
\begin{proof}[Proof of \cref{thm:partial-bases-hardness}]
	Suppose we are given a general linear system (not necessarily clauses of length 3), consisting of $m$ clauses of the form
	\[
		x_{i_1} \oplus x_{i_2} \oplus \cdots \oplus x_{i_\ell} = b,
	\]
	on $n$ binary variables $x_1, \ldots, x_n$, and a threshold $T \in \mathbb N$. To prove our reduction, we need to convert these clauses to an instance of NSPS, with projectors $\pset{110}$, $\pset{111}$, $\pset{210}$, $\pset{211}$, or $\pset{221}$.
	The property these sets have in common (and which will feature in the proof) is that $\ketbra{0}{0}$ and $\ketbra{+}{+}$ are members, but $\ketbra{1}{1}$ is not. We construct an instance of NSPS to be robust to the presence or absence of the other projectors, $\ketbra{-}{-}$, $\ketbra{+i}{+i}$, and $\ketbra{-i}{-i}$.

	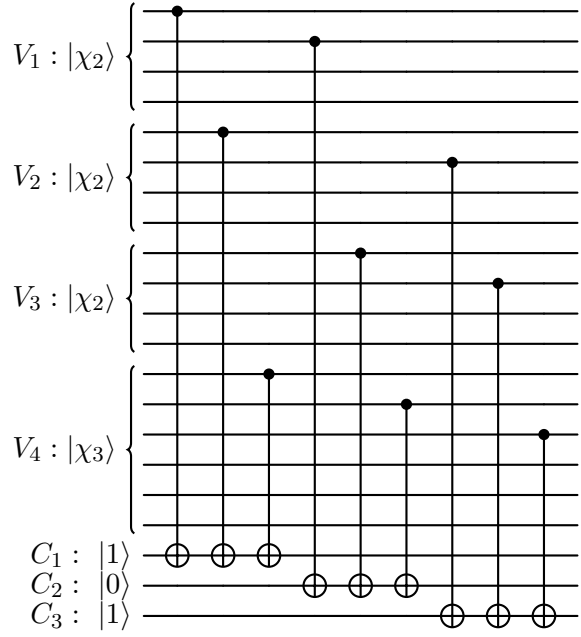
\begin{figure}
		\begin{center}
			\begin{quantikz}[row sep={0.4cm,between origins}, column sep=0.28cm]
				\lstick[4]{$V_1: \ket{\chi_{2}}$} & \ctrl{18} &  &  &  &  &  &  &  &  &  \\
				&  &  &  & \ctrl{18} &  &  &  &  &  &  \\
				&  &  &  &  &  &  &  &  &  &  \\
				&  &  &  &  &  &  &  &  &  &  \\
				\lstick[4]{$V_2: \ket{\chi_{2}}$} &  & \ctrl{14} &  &  &  &  &  &  &  &  \\
				&  &  &  &  &  &  & \ctrl{15} &  &  &  \\
				&  &  &  &  &  &  &  &  &  &  \\
				&  &  &  &  &  &  &  &  &  &  \\
				\lstick[4]{$V_3: \ket{\chi_{2}}$} &  &  &  &  & \ctrl{11} &  &  &  &  &  \\
				&  &  &  &  &  &  &  & \ctrl{11} &  &  \\
				&  &  &  &  &  &  &  &  &  &  \\
				&  &  &  &  &  &  &  &  &  &  \\
				\lstick[6]{$V_4: \ket{\chi_{3}}$} &  &  & \ctrl{6} &  &  &  &  &  &  &  \\
				&  &  &  &  &  & \ctrl{6} &  &  &  &  \\
				&  &  &  &  &  &  &  &  & \ctrl{6} &  \\
				&  &  &  &  &  &  &  &  &  &  \\
				&  &  &  &  &  &  &  &  &  &  \\
				&  &  &  &  &  &  &  &  &  &  \\
				\lstick{$C_1:\ \ket{1}$} & \targ{} & \targ{} & \targ{} &  &  &  &  &  &  &  \\
				\lstick{$C_2:\ \ket{0}$} &  &  &  & \targ{} & \targ{} & \targ{} &  &  &  &  \\
				\lstick{$C_3:\ \ket{1}$} &  &  &  &  &  &  & \targ{} & \targ{} & \targ{} &
			\end{quantikz}
		\end{center}
		\caption{An illustration of the reduction for an instance with clauses $x_1 \oplus x_2 \oplus x_4 = 1$, $x_1 \oplus x_3 \oplus x_4 = 0$, and $x_2 \oplus x_3 \oplus x_4 = 1$.}
		\label{fig:reduction-circuit-diagram}
	\end{figure}

	We build an NSPS instance from this as follows (see \cref{fig:reduction-circuit-diagram} for an example):
	\begin{itemize}
		\item The $j$th clause is associated with one qubit (which we call register $C_j$) in state $\ket{b_j}$. Together, the clause qubits are in state $\ket{b}$.
		\item Each variable $x_i$ is associated with a few qubits (called register $V_i$) in a state of the form
		      \[
			      \ket{\chi_{k_i}} = \frac{\ket{0^{k_i} 1^{k_i}} + \ket{1^{k_i} 0^{k_i}}}{\sqrt{2}},
		      \]
		      where $k_i \in \mathbb N$ is to be determined later. We need $k_i$ to be at least the number of occurrences of $x_i$, so there is a qubit for each occurrence in just the first half of the state.
		\item Let the product of all the variable states be
		      \[
			      \ket{\chi} := \bigotimes_{i} \ket{\chi_{k_i}}.
		      \]
		\item If the $p$th occurrence of variable $x_i$ is in clause $j$, then we add a CNOT gate controlled by qubit $p$ of register $V_i$, and targeting the qubit in $C_j$. Call the product of these CNOTs $U$, a unitary that implements
		      \[
			      U \ket{x} \ket{y} = \ket{x} \ket{Lx \oplus y}
		      \]
		      on the variable and clause qubits, where $L \in \mathbb F^{m \times 2\sum_{i} k_i}$ is some linear transformation.
	\end{itemize}
	The initial states are stabilizer states and the circuit is Clifford, so the resulting state $\ket{\psi}$ is a stabilizer state. We claim that the nearest product state has fidelity
	\begin{equation}
		F = 2^{-\sum_{i} k_i + t - m - n},
	\end{equation}
	where $t$ is the maximum number of clauses that can be satisfied in the $\maxlin$ instance. We will show there is a simple strategy that can achieve at least fidelity $F$, but first we show you cannot do better.

	Suppose some instance (with the minimum number of variables) has fidelity exceeding $F$. If the nearest stabilizer product state for this instance projects onto $\ketbra{0}{0}$ on \emph{any} variable qubit (say, in $V_i$) then that projector commutes through the control of the CNOT gate and acts on a qubit of $\ket{\chi_{k_i}}$. This collapses the state to $\ket{0^{k_i} 1^{k_i}}$ or $\ket{1^{k_i} 0^{k_i}}$ depending on where the projector is.

	Since the initial state of $V_i$ is now effectively a classical basis state, the CNOT gates controlled on these qubits do not create entanglement with the rest of the system. The NSPS is therefore a product of the NSPS for $V_i$, and the NSPS for the rest of the system. The best NSPS for a classical basis state is clear: take $\ketbra{0}{0}$ on $\ket{0}$ qubits and $\ketbra{+}{+}$ on $\ket{1}$ qubits, and thus the fidelity is $2^{-w}$ where $w$ is the Hamming weight. Since both branches of $\ket{\chi_{k_i}}$ have Hamming weight $k_i$, the fidelity is $2^{-k_i-1}$, where the $-1$ accounts for the original $\ketbra{0}{0}$ projector that collapsed the state.

	Moreover, CNOT gates from $V_i$ affect the clause gadgets exactly as if $x_i$ was assigned $0$ (if $V_i$ is in state $\ket{0^{k_i}1^{k_i}}$) or $1$ (if $V_i$ is in state $\ket{1^{k_i} 0^{k_i}}$). Thus, projecting any qubit of $V_i$ onto $\ketbra{0}{0}$ effectively makes an assignment for $x_i$, pays a penalty of $2^{-k_i-1}$ to the fidelity, and reduces to the circuit for a smaller linear system instance (without those variable qubits). It is not hard to see that the fidelity of the full linear system meets the claim if and only if the reduced system meets its respective fidelity claim.

	Notice that if you follow this strategy through with an assignment of all variables, you are left with only clause qubits and, furthermore, each clause is in state $\ket{0}$ if the clause is satisfied and $\ket{1}$ if it is unsatisfied. The NSPS for the clause qubits is a function of the Hamming weight again, so the fidelity is $2^{t-m}$ where $t$ is the number of satisfied clauses. This shows that the claimed fidelity $F = 2^{-\sum_{i} k_i + t - m - n}$ can be achieved with what we'll call the \emph{intended strategy}.

	It follows that our minimal counterexample NSPS must deviate from the intended strategy. In fact, it must project all qubits of all variable registers onto $\pauliX$- and/or $\pauliY$-basis states, otherwise we could reduce to a smaller counterexample. There is a quick calculation that can bound this case for us. The square root of the fidelity is as follows:
	\begin{align*}
		\sqrt{F'} = |\bra{\phi}\!\bra{c} U \ket{\chi}\!\ket{b}|
		 & \leq \sum_{x} \big| (\bra{\phi}\!\bra{c} U \ket{x}\!\ket{b}) \braket{x}{\chi} \big| \\
		 & = \sum_{x} |\braket{\phi}{x} \braket{c}{Lx \oplus b} \braket{x}{\chi}|              \\
		 & \leq 2^{-\sum_{i} k_i} \cdot 2^{-n/2} \cdot 2^{n}                                   \\
		 & = 2^{n/2 - \sum_{i} k_i},
	\end{align*}
	where $x$ is a binary string of length $2\sum_{i} k_i$, to match the length of $\ket{\chi}$. The final inequality comes from three facts:
	\begin{itemize}
		\item $|\braket{c}{Lx \oplus b}| \leq 1$, with equality iff $x$ corresponds to a satisfying assignment.
		\item $\braket{x}{\chi}$ is $0$ for all but $2^n$ cases, corresponding to the assignments of the variables. When it is not zero, we have $|\braket{x}{\chi}| = 2^{-n/2}$.
		\item $|\braket{\phi}{x}| = 2^{-\sum_{i} k_i}$, since we are measuring a length-$2 \sum_{i} k_i$ classical basis string in the $X$- or $Y$-basis.
	\end{itemize}
	After squaring, we conclude that the fidelity of anything outside the intended strategy is at most
	\[
		F' \leq 2^{n-\sum_{i} 2k_i}.
	\]

	As long as $F' \leq F$, the intended strategy gives the optimal NSPS, which we have already showed relates to the optimal solution of the linear system.
	\begin{align*}
		F' \leq 2^{n-2\sum_{i} k_i} & \leq F                             \\
		2^{n-2\sum_{i} k_i}         & \leq 2^{-\sum_{i} k_i + t - m - n} \\
		n-2\sum_{i} k_i             & \leq -\sum_{i} k_i + t - m - n     \\
		2n + m - t                  & \leq \sum_{i} k_i.
	\end{align*}
	We have freedom to pick the $k_i$'s to make this inequality hold. For example, we can take $k_i$ to be two plus the number of occurrences $\occ(x_i)$, of variable $x_i$, so that
	\[
		\sum_{i} k_i \geq \sum_{i} (2 + \occ(x_i)) = 2n + \sum_{i} \occ(x_i) \geq 2n + m,
	\]
	since some variable occurs in every clause. This finishes the proof.
\end{proof}

\subsection{Hardness with full bases}\label{sec:ComplexityCompleteBases}

\begin{theorem}\label{thm:hardness_of_220_and_222}
	$\probdec{220}$ and $\probdec{222}$ are $\NP$-complete.
\end{theorem}

We group $\probdec{220}$ and $\probdec{222}$ together because their associated hardness proofs follow a nearly identical outline. In each case, we will use their solutions to solve the max independent set size problem for graphs. That is, given a graph $G$, our reduction will produce a graph state $\ket\psi$ such that the solution to NSPS for the state $\ket\psi$ allows us to compute $\alpha(G)$, the independence number of $G$. Somewhat surprisingly, this construction (which is given in \Cref{sec:construction}) produces identical graph states for both the $\probdec{220}$ and $\probdec{222}$ problems. In fact, we will treat the $\NP$-hardness proof of $\probdec{220}$ as a warm up to the hardness proof for $\probdec{222}$ since the latter problem must consider strictly more measurement bases.

Since our construction produces a graph state, it is worth taking some time to phrase $\probdec{220}$ and $\probdec{222}$ in more graph-theoretic language. To do this, let us first set some terminology. Let $G = (V,E)$ be a simple graph. For any subset $W \subseteq V$, we use $G[W]$ for the induced subgraph on the vertices $W$. We will use the graph and its adjacency matrix interchangeably, so that, e.g., $\rank(G)$ is the rank of the adjacency matrix (over $\mathbb F_2$). Define the nullity of $G$ as $\eta(G) := |V| - \rank(G)$.

\begin{lemma}
	\label{lem:NSPS_to_nullity}
	Let $\ket\psi$ be the graph state of $G = (V,E)$. Then
	\begin{align}
		-\log_2 m(\psi,\pset{220}) & =\abs{V} - \max_{W \subseteq V} \eta(G[W]) \label{eq:nullity220}                         \\
		-\log_2 m(\psi,\pset{222}) & =\abs{V} - \max_{W \subseteq V, \; D \in \diag{|W|}} \eta(G[W] + D)\label{eq:nullity222}
	\end{align}
	where $\diag{m}$ is the set of diagonal $m \times m$ matrices over $\mathbb F_2$.
\end{lemma}

In other words, solutions to the $\probdec{220}$ and $\probdec{222}$ problems correspond to solutions to nullity maximization problems over graphs. For  $\probdec{220}$, we want to maximize the nullity of an induced subgraph of $G$. For $\probdec{222}$, we want to maximize the nullity of an induced subgraph of $G$ where we additionally are allowed to add entries along the diagonal of the adjacency matrix.

\begin{proof}[Proof of \Cref{lem:NSPS_to_nullity}]
	The proof is a straightforward interpretation of the Gottesman--Knill algorithm on the stabilizer tableau of a graph state. We adopt the notation $[C \mid D]$ to refer to the Pauli generators of some stabilizer state with $X$-components given by the matrix $C$ and $Z$-components given by the matrix $D$ (as in Aaronson--Gottesman \cite{aaronson2004improved}). Given a graph state with adjacency matrix $A$, it is well-known that its stabilizer generators have tableau $[ I \mid A ]$, where $I$ is the identity matrix.

	From \cref{lem:200_UB}, we know how to maximize the measurement probability when all qubits are measured in the Pauli-$Z$ basis. To measure a qubit in the $X$ basis (or $Y$ basis), we can first apply a Hadamard gate (or a phase gate followed by a Hadamard gate) to the qubit, then measure in the Pauli-$Z$ basis. Let us consider how the $X$-component of the tableau is modified under these operations:
	\begin{itemize}
		\item Pauli-$Z$ measurement of $j^{\mathrm{th}}$ qubit: the $j^{\mathrm{th}}$ column remains unchanged and corresponds to the standard basis column vector $\mathrm{col}_j(I)=e_j$;
		\item Pauli-$X$ measurement of $j^{\mathrm{th}}$ qubit: the $j^{\mathrm{th}}$ column of $C$ and $D$ are swapped resulting in $j^{\mathrm{th}}$ column of the adjacency matrix $\mathrm{col}_j(A)=:a_j$;
		\item Pauli-$Y$ measurement of $j^{\mathrm{th}}$ qubit: the $j^{\mathrm{th}}$ column of $C$ and $D$ are added resulting in $a_j+e_j$.
	\end{itemize}
	Hence, the maximization of the fidelity over $\pset{220}$ is equivalent
	to the minimization of the rank of the $n \times n$ matrix constructed by choosing either column $e_j$ or $a_j$ for each $j\in [n]$. Similarly, the maximization of the fidelity over $\pset{222}$ is equivalent
	to the minimization of the rank of the $n \times n$ matrix constructed by choosing column $e_j$ or $a_j$ or $a_j+e_j$ for each $j\in [n]$.

	Let the sets $F\subseteq W \subseteq [n]$ denote our choices of measurement basis as follows. Let $W$ be the set of qubits that we measured in either the $X$ or $Y$ basis, and let $F$ be the subset of $W$ for the qubits that we measured in the $Y$ basis. The $\pset{220}$ setting corresponds to a rank minimization over $W$ subject to $F=\emptyset$ while the $\pset{222}$ setting corresponds to a rank minimization over both $F$ and $W$.

	Let us now consider the form of the $X$-component of the final tableau (prior to $Z$ basis measurement) for a given choice of $F$ and $W$. Since the rank of a matrix is invariant under the permutation of its rows and columns, we can apply such permutations to collect all the 1's associated with qubits in $W^c:=[n]\setminus W$ (i.e., the qubits measured in the Pauli-$Z$ basis) to the top left of the matrix. This results in a final matrix of the form:
	\begin{align}
		C(F,W):=\begin{bmatrix} I & J \\0 &A[W]+D(F)
		        \end{bmatrix},
	\end{align}
	where $I$ is an identity matrix of size $\abs{W^c}=n-\abs{W}$, $A[W]$ is the principal submatrix of $A$ where rows and columns in $W^c$ are deleted, $D(F)$ is a diagonal matrix of size $\abs{W}$ with the diagonal entries in $D(F)$ set to $1$ if the associated qubit is in $F$ and otherwise set to $0$. The matrix $J$ is unimportant for this proof but is associated with edges of $G$ between a qubit in $W$ and a qubit in $W^c$.

	Hence:
	\begin{align}
		-\log_2 m(\psi,\pset{222}) & =\min_{F \subseteq W \subseteq [n]} \rank C(F,W)                               \\
		                           & =\min_{F \subseteq W \subseteq [n]} (\abs{W^c}+ \rank \left[ A[W]+D(F)\right]) \\
		                           & =\min_{F \subseteq W \subseteq [n]} (n-\abs{W}+\rank \left[ A[W]+D(F)\right])  \\
		                           & =n+\min_{F \subseteq W \subseteq [n]} -\eta\left( A[W]+D(F)\right)             \\
		                           & =n-\max_{F \subseteq W \subseteq [n]} \eta\left( A[W]+D(F)\right).
	\end{align}
	The \cref{eq:nullity222} claim is immediate. From this, we also attain the \cref{eq:nullity220} claim by enforcing $F=\emptyset$ (disallowing $Y$ measurements).
\end{proof}

\subsubsection{The graph construction}
\label{sec:construction}

We will use the same reduction for each problem from independent set.

Let $G = (V, E)$ be a simple graph. We will construct a new graph $H$ by replacing each edge $e = (u, v) \in E$ with a path of length 3. That is, for each edge $e \in E$, we introduce two \emph{internal} vertices $p_e$ and $q_e$, replacing $e$ with the edges $(u, p_e)$, $(p_e, q_e)$, and $(q_e, v)$:

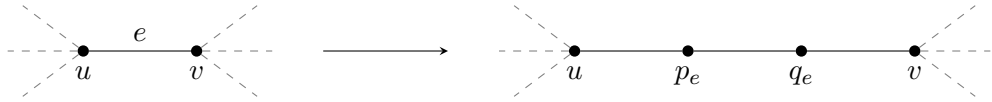
\begin{center}
	\begin{tikzpicture}[
			vertex/.style={circle, fill=black, inner sep=1.5pt, outer sep=0pt},
			dangling/.style={thin, dashed, gray},
			edge/.style={thin}
		]
		\newcommand{\edgelen}{1.5}

		\node[vertex, label=below:$u$] (u1) at (0,0) {};
		\node[vertex, label=below:$v$] (v1) at (\edgelen,0) {};

		\draw[edge] (u1) -- (v1) node[midway, above] {$e$};

		\draw[dangling] (u1) -- +(-0.8, 0.6);
		\draw[dangling] (u1) -- +(-1.0, 0);
		\draw[dangling] (u1) -- +(-0.8, -0.6);

		\draw[dangling] (v1) -- +(0.8, 0.6);
		\draw[dangling] (v1) -- +(1.0, 0);
		\draw[dangling] (v1) -- +(0.8, -0.6);

		\node[vertex, label=below:$u$] (u2) at (6.5,0) {};
		\node[vertex, label=below:$p_e$] (pe) at ({6.5+\edgelen},0) {};
		\node[vertex, label=below:$q_e$] (qe) at ({6.5 + 2*\edgelen},0) {};
		\node[vertex, label=below:$v$] (v2) at ({6.5 + 3*\edgelen},0) {};

		\draw[edge] (u2) -- (pe) -- (qe) -- (v2);

		\draw[dangling] (u2) -- +(-0.8, 0.6);
		\draw[dangling] (u2) -- +(-1.0, 0);
		\draw[dangling] (u2) -- +(-0.8, -0.6);

		\draw[dangling] (v2) -- +(0.8, 0.6);
		\draw[dangling] (v2) -- +(1.0, 0);
		\draw[dangling] (v2) -- +(0.8, -0.6);

		\draw[-stealth, shorten <=1.6cm, shorten >=1.6cm] (v1) -- (u2);

	\end{tikzpicture}
\end{center}

Let $P = \{ p_e : e \in E \} \cup \{q_e : e \in E\}$ be the set of these extra path vertices. That is, the vertex set of $H$ is $V(H) = V \cup P$.

\begin{lemma}
	\begin{align*}
		\max_{W \subseteq V(H)} \eta(H[W]) = \max_{W \subseteq V(H), \; D \in \diag{|W|}} \eta(H[W] + D) = \alpha(G) + |E|.
	\end{align*}
\end{lemma}

We show that the maximum is achieved by choosing $W$ such that:
\begin{itemize}
	\item $W \cap V$ is a maximum independent set for $G$ and
	\item $W\setminus V$ contains exactly one vertex from each pair $(p_e,q_e)$ for all $e\in E$ subject to the constraint that no vertex from $W \cap V$ is adjacent to a vertex from $W\setminus V$.
\end{itemize}
This is always possible since the independent set requirement ensures that no edge in $E$ has both endpoints in $W \cap V$. By construction, this choice of $W$ is an independent set of $H$ of size $\alpha(G) + |E|$. We proceed by first showing that $\alpha(G) + |E|$ lower bounds the maximization above since the above choice of $W$ achieves this nullity. Then, we conclude by showing that all choices of $W$ must result in a nullity that is at most $\alpha(G) + |E|$.

The lower bound in each case is immediate. Since the adjacency matrix of an independent set is the all-zeros matrix, the nullity is equal to the dimension, i.e., $\abs{W}=\alpha(G)+\abs{E}$.

\subsubsection{\texorpdfstring{$\probdec{220}$}{NSPS(220)}}

We must show that $\eta(H[W]) \le \alpha(G) + |E|$ for any induced subgraph $H[W]$. We start by partitioning the edges of $G$ into three sets based on their internal vertices in $W$: for each $i \in \{0,1,2\}$,
\[
	E_i = \{ e \in E : |\{p_e, q_e\} \cap W| = i \}
\]
Let $W_G = W \cap V$ be the vertices in $H[W]$ that are also in $G$. Using these definitions, we get that the total number of vertices in $H[W]$ is $|W| = |W_G| + |E_1| + 2|E_2|$.

Much of the analysis follows by a sequence of row/column reductions on the adjacency matrix for $H[W]$, so let us pick an order for the vertices $W$:
\begin{enumerate}[itemsep = 0pt]
	\item The vertices in $W_G$, followed by
	\item The internal vertices from $E_1$, followed by
	\item The internal vertices from $E_2$ with adjacent vertices next to each other
\end{enumerate}
In other words, the adjacency matrix of $H[W]$ takes the block form:
\[
	\begin{pmatrix} 0 & B & C  \\ B^T & 0 & 0 \\ C^T & 0 & D  \end{pmatrix}
\]
where $B$ and $C$ are arbitrary, and $D$ is a $2|E_2| \times 2|E_2|$ block matrix with blocks $\left( \begin{smallmatrix} 0 & 1 \\ 1 & 0 \end{smallmatrix} \right)$ along the diagonal reflecting the edges between neighboring vertices of $E_2$. The other zero submatrices appear because $V$ is an independent set in $H$ and the vertices of $E_1$ are isolated from each other.

Importantly, the matrix $D$ is invertible (and, in fact, involutory $D^2 = I$), so we can reduce the matrix into two principal submatrices through some simple row/column reductions:
\[
	\begin{pmatrix} I & 0 & C D \\ 0 & I & 0 \\ 0 & 0 & I \end{pmatrix}
	\begin{pmatrix} 0 & B & C \\ B^T & 0  & 0 \\ C^T & 0 & D \end{pmatrix}
	\begin{pmatrix} I & 0 & 0 \\ 0 & I & 0 \\  D C^T & 0 & I \end{pmatrix}
	= \begin{pmatrix} C D C^T & B & 0 \\ B^T & 0 & 0 \\ 0 & 0 & D \end{pmatrix}
\]

In other words, we can compute the rank of the adjacency matrix of the induced graph $H[W]$ as the sum of the ranks of the two principal submatrices:
\[
	\rank(H[W]) = \rank\begin{pmatrix} C D C^T & B \\ B^T & 0 \end{pmatrix}  + \rank D = \rank\begin{pmatrix} C D C^T & B \\ B^T & 0 \end{pmatrix}  + 2|E_2| \ge  2 \rank(B) + 2 |E_2|.
\]
Interestingly, $C D C^T$ is the adjacency matrix of a subgraph of $G$ with vertex set $W_G$ and edge set $E_2$. This may give some intuition for the construction since choosing two neighboring internal vertices in $E_2$ (which, by definition, will not be part of an independent set) will recreate the corresponding edge of $G$ in $C D C^T$, which affects the rank. That said, we will not actually need or use this interpretation.

Our rank calculation gives a straightforward bound on the nullity of $H[W]$:
\[
	\eta(H[W]) = |W| - \rank(H[W]) \le (|W_G| + |E_1| + 2|E_2|) - (2\rank(B) + 2|E_2|)
\]
simplifying to
\[
	\eta(H[W]) \le |W_G| + |E_1| - 2\operatorname{rank}(B).
\]
Recall that the matrix $B$ has $|E_1|$ columns, where each column corresponds to a single internal vertex from an edge in $E_1$. Since this internal vertex connects to exactly one endpoint $v \in V$, its column is either the standard basis vector (if $v \in W_G$) or the zero vector (if $v \notin W_G$). Therefore, the rank of $B$ is exactly the number of \emph{distinct} basis vectors it contains. Let $U \subseteq W_G$ be this set of distinct vertices, so $\rank(B) = |U|$. Ultimately, we will show that the nullity is maximized when $U = \emptyset$, reflecting the fact that an independent set is the best choice.

To that end, let's substitute $|W_G| = |W_G \backslash U| + |U|$ into our bound for the nullity:
\[
	\eta(H[W]) \le (|W_G \backslash U| + |U|) + |E_1| - 2|U| = |W_G \backslash U| + |E_1| - |U|
\]
and so, our goal is to now bound the size of $W_G \backslash U$. Our plan will be to use Gallai's lemma to relate the size $W_G \backslash U$ to the size of its largest independent set. In particular, for the induced subgraph $G[W_G \backslash U]$, the sum of its independence number $\alpha(W_G \backslash U)$ and its vertex cover number $\tau(W_G \backslash U)$ equals its vertex count:
\[
	|W_G \backslash U| = \alpha(W_G \backslash U) + \tau(W_G \backslash U).
\]
Clearly, the max independent set size of an induced subgraph is at most the max independent set of the full graph: $\alpha(W_G \backslash U) \le \alpha(G)$ (which, recall, is part of the target bound on the nullity). To bound the size of the minimum vertex cover, we will use a standard trick---the size of the minimum vertex cover is at most the number of edges since we can choose one vertex of each edge to be part of the cover. This immediately gives the bound $\tau(W_G \backslash U) \le |E(W_G \backslash U)| \le |E|$.

Unfortunately, this bound isn't quite tight enough for our needs. Thankfully, we can get a tighter bound using the fact that $U$ contains all those vertices of $W_G$ that are adjacent to an internal vertex from $E_1$. Consequently, it is impossible for an edge in $E_1$ to have both of its endpoints in $W_G \backslash U$ (since the endpoint adjacent to the internal vertex is either in $U$ or not in $W_G$). It follows that:
\[
	|E(W_G \backslash U)| \le |E| - |E_1|
\]
Plugging back into Gallai's lemma, we get
\[
	|W_G \backslash U| = \alpha(W_G \backslash U) + \tau(W_G \backslash U) \le \alpha(G) + |E| - |E_1|
\]
and, substituting into our bound for the nullity yields
\[
	\eta(H[W]) \le (\alpha(G) + |E| - |E_1|) + |E_1| - |U| = \alpha(G) + |E| - |U| \le \alpha(G) + |E|
\]
which concludes the proof.

\subsubsection{\texorpdfstring{$\probdec{222}$}{NSPS(222)}}

As stated previously, the reduction for $\probdec{222}$ is identical. However, the analysis is more involved. If not otherwise stated, we adopt all of the same notation from that proof. Adopting the same ordering on the vertices, we now seek to understand the rank of the matrix
\[
	A := H[W] + \mathsf{Diag} = \begin{pmatrix} 0 & B & C  \\ B^T & 0 & 0 \\ C^T & 0 & D  \end{pmatrix} + \mathsf{Diag}
\]
where $\mathsf{Diag}$ is an arbitrary diagonal matrix. The extra diagonal matrix will require us to subdivide the vertices in $E_1$ and $E_2$ by the values along the diagonal. For each internal vertex from $E_1$, there is only one binary choice---either there is a 1 along the diagonal or there is not. For each pair of internal vertices in $E_2$, there are now $4$ possible principal submatrices
\begin{itemize}
	\item Non-invertible: $\left(\begin{smallmatrix} 1 & 1 \\ 1 & 1 \end{smallmatrix} \right)$
	\item Invertible: $\left(\begin{smallmatrix} 0 & 1 \\ 1 & 0 \end{smallmatrix} \right)$, $\left(\begin{smallmatrix} 1 & 1 \\ 1 & 0 \end{smallmatrix} \right)$, $\left(\begin{smallmatrix} 0 & 1 \\ 1 & 1 \end{smallmatrix} \right)$
\end{itemize}

For $i = 1,2$, we define the \emph{singular} edges $S_i \subseteq E_i$ to be the edges in $E_i$ that have non-invertible blocks on the diagonal. The vertices are now ordered starting with $W_G$ followed by the internal vertices from $S_1$, $E_1 \backslash S_1$, $S_2$, and $E_2 \backslash S_2$. In other words, we must analyze a matrix of the form
\[
	\begin{pmatrix}
		\Lambda & B_0 & B_1 & C_0 & C_1 \\
		B_0^T   & 0   & 0   & 0   & 0   \\
		B_1^T   & 0   & I   & 0   & 0   \\
		C_0^T   & 0   & 0   & D_0 & 0   \\
		C_1^T   & 0   & 0   & 0   & D_1
	\end{pmatrix}
\]
where $\Lambda$ is an arbitrary diagonal matrix, $D_0$ is block diagonal with the non-invertible $2 \times 2$ submatrices above, and $D_1$ is block diagonal with the invertible $2 \times 2$ submatrices. Once again, we will use the invertible diagonal entries to row/column reduce our matrix.
So,
\[
	\eta(A) = \eta
	\begin{pmatrix}
		M     & B_0 & C_0 \\
		B_0^T & 0   & 0   \\
		C_0^T & 0   & D_0
	\end{pmatrix}
\]
where $M = \Lambda + B_1 B_1^T + C_1 D_1^{-1} C_1^T$. Recall that $D_0$ is block diagonal with the blocks $\left(\begin{smallmatrix} 1 & 1 \\ 1 & 1 \end{smallmatrix} \right)$. Let's focus on a principal submatrix corresponding to one of the blocks corresponding to the edge $e$:
\[
	\begin{pmatrix}
		M     & a_e & b_e \\
		a_e^T & 1   & 1   \\
		b_e^T & 1   & 1
	\end{pmatrix}
\]
where $a_e, b_e$ are column vectors of length $|W_G|$. Once again, our plan will be to isolate a principal submatrix of full rank that can be removed:
\[
	\begin{pmatrix}
		I & 0 & b_e \\
		0 & 1 & 1   \\
		0 & 0 & 1
	\end{pmatrix}
	\begin{pmatrix}
		M     & a_e & b_e \\
		a_e^T & 1   & 1   \\
		b_e^T & 1   & 1
	\end{pmatrix}
	\begin{pmatrix}
		I     & 0 & 0 \\
		0     & 1 & 0 \\
		b_e^T & 1 & 1
	\end{pmatrix}
	=
	\begin{pmatrix}
		M + b_e b_e^T & a_e + b_e & 0 \\
		a_e^T + b_e^T & 0         & 0 \\
		0             & 0         & 1
	\end{pmatrix}
\]
Note that $a_e + b_e$ exactly encodes the incidence between the vertices in $W_G$ and the edge $e$.

To summarize, after once again removing the full rank submatrices, we get
\[
	\eta(A) = \eta
	\begin{pmatrix}
		\widetilde{M} & B_0 & \Delta \\
		B_0^T         & 0   & 0      \\
		\Delta^T      & 0   & 0
	\end{pmatrix}
\]
where $\widetilde{M} = M + \sum_e b_e b_e^T$ and $\Delta$ is the incidence matrix for the vertices $W_G$ and edges $S_2$.

We are now in position to repeat large parts of the earlier proof, first, by bounding the rank of the matrix above:
\[
	\rank\begin{pmatrix}
		\widetilde{M} & B_0 & \Delta \\
		B_0^T         & 0   & 0      \\
		\Delta^T      & 0   & 0
	\end{pmatrix} \ge  2 \rank(\begin{pmatrix} B_0  & \Delta \end{pmatrix})
\]
so that
\[
	\eta(A) \le |W_G| + |S_1| + |S_2| - 2 \rank(\begin{pmatrix} B_0  & \Delta \end{pmatrix})
\]
To bound the rank we analyze the columns of $B_0$ and $\Delta$. Once again, we will construct a set of vertices $U \subseteq W_G$ such that its corresponding submatrix in $(B_0 \;\; \Delta)$ has full rank. Recall that each column of $B_0$ is either the all-zeroes vector or a standard basis vector.  Similarly, each column of $\Delta$ has zero, one, or two $1$'s in it.  Add all the vertices in $W_G$ corresponding to one of the standard basis columns to $U$. After removing these rows from $(B_0 \;\; \Delta)$, there may be a new set of columns that become standard basis vectors, so we add them to $U$ as well, and repeat. That is, we perform Gaussian elimination on the columns of $(B_0 \;\; \Delta)$ using standard basis vectors as the pivots until each column is either the all-zeros vector or has two $1$'s. The columns with two $1$'s correspond to edges in $S_2$ with both endpoints in $W_G \backslash U$.

So, in summary $\rank(B_0  \; \Delta)$ is $|U|$ plus the rank of the incidence matrix of $W_G \backslash U$ with the edges in $S_2$. We claim that the rank of this incidence matrix is the number of edges in any spanning forest $F$ of the graph with vertex set $W_G  \backslash U$ and edge set $S_2$. To see this, it suffices to note that any column corresponding to edge (say, $\mathsf{col}(u,v)$) is the sum of the columns forming a path from $u$ to $v$:
\[
	\mathsf{col}(u,v) = \mathsf{col}(u, x_1) + \mathsf{col}(x_1, x_2) + \ldots + \mathsf{col}(x_{k-1}, x_k) + \mathsf{col}(x_k, v)
\]
for edges $(x_i, x_{i+1})$ in the subgraph of vertices in $W_G \backslash U$ with edges in $S_2$. Therefore, we get $\rank(B_0  \; \Delta) = |U| + |E(F)|$.

Plugging back into our nullity bound, we get
\[
	\eta(A) \le |W_G \backslash U| + |U| + |S_1| + |S_2| - 2 ( |U| + |E(F)| ) = |W_G \backslash U| + |S_1| + |S_2| - |U| - 2|E(F)|
\]
Once again, we will use Gallai's lemma to bound $|W_G \backslash U|$, but we will need to be more careful with how we construct the vertex cover. To construct the cover, first notice that the edges of $G[W_G \backslash U]$ can be partitioned into four sets: the edges in $E_0$, the edges in $E_1 \backslash S_1$, the edges in $S_2$, and the edges in $E_2 \backslash S_2$. Recall that no edge in $S_1$ can appear in $G[W_G \backslash U]$.

We choose one vertex from each edge in $E_0$, $E_1 \backslash S_1$, and $E_2 \backslash S_2$ to cover those edges. To cover the edges in $S_2$, we choose all-but-one vertex in each connected component of the spanning forest. Altogether, we get a vertex cover of size
\[
	|E_0| + |E_1 \backslash S_1| + |E_2| - |S_2| + |E(F)|.
\]
Plugging into Gallai's lemma (using $\alpha(W_G \backslash U) \le \alpha(G)$ again), we get
\[
	|W_G \backslash U| = \alpha(W_G \backslash U) + \tau(W_G \backslash U) \le \alpha(G) + |E_0| + |E_1 \backslash S_1| + |E_2| - |S_2| + |E(F)|
\]
so, we get
\begin{align*}
	\eta(A) & \le (\alpha(G) + |E_0| + |E_1 \backslash S_1| + |E_2| - |S_2| + |E(F)|) + |S_1| + |S_2| - |U| - 2|E(F)| \\
	        & \le \alpha(G) + (|E_0| + |E_1| + |E_2|) - |U| - |E(F)|                                                  \\
	        & = \alpha(G) + |E| - |U| - |E(F)|                                                                        \\
	        & \le \alpha(G) + |E|
\end{align*}
completing the proof. \qed

\section*{Statement on AI Use}
Significant parts of the hardness proofs for $\probdec{220}$ and $\probdec{222}$ were found by an interactive prompting of Gemini 3.1 Pro. The proofs were read, understood, and rewritten by the authors, so little to no AI generated text remains. The $\probdec{1{*}{*}}$ hardness results are entirely human created. Gemini 3.1 Pro was also used to search for relevant results on various topics. Claude Opus 5.5 was used to check for correctness, bugs, and typos.

\section*{Acknowledgments}
HP thanks Oliver Reardon-Smith and Kamil Korzekwa for useful discussions during the early stages of this project. HP acknowledges working on this project during previous employments at
Dahlem Center for Complex Quantum Systems, Freie Universitat Berlin, Institute for Quantum Computing, University of Waterloo and Perimeter Institute for Theoretical Physics.
LS is supported under NSERC RGPIN-2025-04875.

\bibliographystyle{plain}
\bibliography{bibliography}

\appendix
\crefalias{section}{appendix}

\section{\textsc{NSPS} Variant Definitions and Reductions}
\label{app:FormalDefin}

We define the decision and search problem variants associated with the optimization problem from \cref{defn:probnum}.

\begin{definition}[$\probdec{\stab{n},\sscase}$]\label{defn:probdec}
	For $\sscase\subseteq \sall$ such that $\sscase \neq \emptyset$, we define the decision problem $\probdec{\stab{n},\sscase}$ as follows:

	Given any state $\ket{\psi}\in \stab{n}$, $k\in [n]$ and $p\in\bbn$, decide if:
	\begin{align*}
		 & \pi_{YES}: m({\psi},\Pi_{\sscase,k}) \geq 2^{-p+1} \\
		 & \pi_{NO}: m({\psi},\Pi_{\sscase,k}) \leq 2^{-p}.
	\end{align*}
\end{definition}

\begin{definition}[$\probproj{\stab{n},\sscase}$]\label{defn:probproj}
	For $\sscase\subseteq \sall$ such that $\sscase \neq \emptyset$, we define the search problem $\probproj{\stab{n},\sscase}$ as follows:

	Given any state $\ket{\psi}\in \stab{n}$ and $k\in [n]$, output $\pi \in \Pi_{\sscase,k}$ satisfying:
	\begin{align}\label{eq:opt_projector}
		\twonorm{\pi \ket{\psi}}^2= m({\psi},\Pi_{\sscase,k})
	\end{align}
\end{definition}

We show that for any fixed subset of projectors $\sscase$,
given the ability to access at most $\poly(n)$ solutions to desired instances of any one of the problems $\probdec{\stab{n},\sscase}$, $\probnum{\stab{n},\sscase}$ or $\probproj{\stab{n},\sscase}$, one can solve any given instance of the other two variants. We state this more concisely by using the following notation: for two problems $A$ and $B$, we write $A \reducesto B$ to denote that there exists a polynomial-time Turing reduction from $A$ to $B$ and hence solving $A$ is possible in $\poly(n)$ time with access to a solver for $B$.

\begin{lemma}
	For any fixed $\sscase\subseteq \sall$, $\sscase\neq \emptyset$, all of the following hold:
	\begin{align}
		\probproj{\stab{n},\sscase} & \reducesto \probnum{\stab{n},\sscase}, \label{eq:reduction_proj-num}  \\
		\probnum{\stab{n},\sscase}  & \reducesto \probdec{\stab{n},\sscase}, \label{eq:reduction_num-dec}   \\
		\probdec{\stab{n},\sscase}  & \reducesto \probproj{\stab{n},\sscase}. \label{eq:reduction_dec-proj}
	\end{align}
\end{lemma}

\begin{proof}
	\cref{eq:reduction_dec-proj} is immediately implied by the Gottesman--Knill theorem.
	To show \cref{eq:reduction_proj-num} and \cref{eq:reduction_num-dec}, let us fix $\sscase \subseteq \sall$ with $\abs{S}=q$ and assume we are interested in solving some instance $\psi_{\emptyset} \in \stab{n}$ and $k\in [n]$. The case where $q<2$ is trivial so we can assume $q\geq 2$.

	We first show \cref{eq:reduction_num-dec} holds. For $q\geq 2$
	, $- \log_2 m(\psi_{\emptyset}, \Pi_{\sscase,k}):=p^* \in \{ 0, \ldots, k \}$. Our strategy will simply be to do a binary search over the set $\{ 0, \ldots, k \}$ by calling the oracle for $\probdec{\stab{n},\sscase}$ on the same input state but different inputs of $p$ from \cref{defn:probdec}. This requires $\order{\log_2 k}$ calls to the oracle.

	We now show \cref{eq:reduction_proj-num} holds. The solution to the search problem involves a search over $q^{k}$ projectors to find any projector $\pi \in \Pi_{\sscase,k}$ that satisfies \cref{eq:opt_projector}. The projector satisfying this equation may not be unique however, given access to an oracle for $\probnum{\stab{n},\sscase}$, we can conduct a search over the complete $q$-ary tree to find a solution to the search problem. This uses $\order{qk}=\order{k}$ calls to the oracle for the optimization problem. For simplicity of notation, when the state $\ket{\psi'}=a \ket{\psi}$ is unnormalized with $\ket{\psi}$ a normalized stabilizer state and $a\in \bbr$, we will use $m(\psi',\Pi_{\sscase,k})$ to mean $a^2~ m(\psi,\Pi_{\sscase,k})$.

	We label each node in the complete $q$-ary tree with $x\in \set{\emptyset} \cup \bbz_q \cup \bbz_q^2 \cup \ldots \cup \bbz_q^k$ such that the root node has label $x=\emptyset$ and nodes at level $j$ have labels $x\in \bbz_q^j$ assigned in the natural way. That is, the children of node $x=(x_1, \ldots, x_j)$ have labels $(x_1, \ldots, x_j, a)$ for $a\in \bbz_q$. Each node $x$ is associated with an (unnormalized) stabilizer state $\ket{\psi_x}$.
	By labeling the distinct elements of $S$ with $\pi^{(a)}$ for $a\in \bbz_q$, we define all states at level $j$ using
	\begin{align}
		\ket{\psi_x} & :=\pi^{(x_1)}\otimes \ldots \otimes \pi^{(x_j)} \otimes I^{\otimes n-j} \ket{\psi_{\emptyset}} \\
		             & = \pi(x) \ket{\psi_{\emptyset}}.\label{eq:psi_x_form}
	\end{align}

	We now associate with each state $\ket{\psi_x}$ the `$m$-value' $m_x:=m(\psi_x, \Pi_{\sscase,k})$. It is easy to see that as we traverse the tree from root to leaf:
	\begin{enumerate}
		\item the $m$-value associated with the nodes is non-increasing
		\item at each step from parent to one of the children, there is always a choice of child that preserves the $m$-value.
	\end{enumerate}
	Hence, finding the solution of the search problem is only a matter of starting at the root and walking along the tree via any path with constant $m$-value.

	To establish the first property, we simply observe that for any $\pi_{loc} \in \sscase$ and  stabilizer state $\ket{\psi}$, we have $\twonorm{[I^{\otimes j-1} \otimes \pi_{loc} \otimes I^{\otimes n-j}] \ket{\psi}}^2\in \set{0,1/2,1} \cdot \twonorm{\ket{\psi}}^2$.

	To establish the second property, observe that by definition of $m_x$, there exists $\pi\in \Pi_{\sscase,k}$ such that $\twonorm{\pi \ket{\psi_x}}^2=m_x$. By \cref{eq:psi_x_form} and the optimality of $\pi$ we have, $\pi \pi(x)=\pi=\pi^{(x_1)}\otimes \ldots \otimes \pi^{(x_k)} \otimes I^{\otimes n-k}$ for some choice of $x_{j+1},\ldots, x_k$. Since the value of $x_{j+1}\in \bbz_q$ corresponds to one of the children of node $(x_1, \ldots, x_j)$ we have $m_{(x_1, \ldots, x_j)}=m_{(x_1, \ldots, x_j, x_{j+1})}$.

\end{proof}

\section{Classical Simulation of Clifford Plus T Circuits}\label{app:additive_CSim}

Below, we give an aggressively simplified account of the $\textsc{RawEstim}$ algorithm of Ref.~\cite{pashayan2022fast} that nevertheless, outlines the mathematical connection to $\probnum{110}$. This simulation algorithm starts with the established approach \cite{bravyi2016improved, howard2017application} of using gate teleportation \cite{zhou2000methodology} to replace the $T$-gates with $T$-states  ($\ket{T}=T\ket{+}$ state).  If the $n$-qubit target circuit $U$ originally contained $t$ occurrences of the $T$-gate, this gate teleportation procedure results in a Clifford circuit $C$ acting on $\ket{0}^{\otimes n} \otimes \ket{T}^{\otimes t}$. The Born rule probability $p$
of observing the $0^n$ outcome when the original circuit $U \ket{0^n}$ is measured in the computational basis,
can be expressed in terms of the Born rule probability $p'$ of observing the $0^{n+t}$ outcome when the new circuit is measured in the computational basis using $p=2^t p'$. The goal then becomes estimating
\begin{align}\label{eq:Born_prob_gadgetized_circuit}
	p'=\abs{\bra{0^{n+t}} C \ket{0^n T^{t}}}^2 = \abs{\braket{\psi}{T^{t}}}^2
\end{align}
to sufficiently high precision ($\pm 2^{-t}\epsilon$), where $\ket{\psi}= (\bra{0^n} \otimes I^{\otimes t}) C^{\dagger} \ket{0^{n+t}}$ is a (sub-normalized) stabilizer state and we have used the concise notation $\ket{0^n T^{t}}:=\ket{0}^{\otimes n} \otimes \ket{T}^{\otimes t}$. Ref.~\cite{pashayan2022fast} proceeds by decomposing the magic states $\ket{T^t}$ into a linear combination of stabilizer states, which can be simulated separately. The 1-norm of the stabilizer coefficients in this decomposition is the key driver of simulation runtime. The square of this 1-norm is minimized (achieving the so-called stabilizer extent $\xi$) when the magic states are decomposed into a linear combination of $\ket{+}$ and $\ket{+i}$ states as follows:
\begin{align}
	\ket{T^t} & =\sqrt{\xi} \sum_{\ket{\phi}\in \set{\ket{+},\ket{+i}}^{\otimes t}} \frac{e^{i\theta(\phi)}}{2^t} \ket{\phi}                                    \\
	          & =\sqrt{\xi} \underset{\ket{\phi}\sim \mathrm{Unif} \set{\ket{+},\ket{+i}}^{\otimes t}}{\mathbb{E}} \left[ e^{i\theta(\phi)} \ket{\phi} \right],
\end{align}
where $\theta(\phi)$ are some phases that will not be important here. The estimation of the inner product:
\begin{align}
	\braket{\psi}{T^{t}}=\sqrt{\xi} \underset{\ket{\phi}\sim \mathrm{Unif} \set{\ket{+},\ket{+i}}^{\otimes t}}{\mathbb{E}} \left[ e^{i\theta(\phi)} \braket{\psi}{\phi}\right]
\end{align}
is achieved by using the mean of $\sqrt{\xi}e^{i\theta(\phi)} \braket{\psi}{\phi}$ based on many iid samples of $\ket{\phi}$. Using the Hoeffding inequality, the number of required samples, $N$, can be upper bounded by the square of the modulus:
\begin{align}\label{eq:csim_num_samples_UB}
	N \leq \order{\epsilon^{-2} \underset{\ket{\phi} \in \set{\ket{+},\ket{+i}}^{\otimes t}}{\max} \abs{\sqrt{\xi} e^{i\theta(\phi)} \braket{\psi}{\phi}}^2}
	= \order{\xi \epsilon^{-2} \underset{\ket{\phi} \in \set{\ket{+},\ket{+i}}^{\otimes t}}{\max} \abs{\braket{\phi}{\psi}}^2}.
\end{align}
The RHS provides the direct link between $\probnum{110}$ and the performance of this simulator. In Ref.~\cite{pashayan2022fast}, it was noted that
\begin{align*}
	\underset{\ket{\phi} \in \set{\ket{+},\ket{+i}}^{\otimes t}}{\max} \abs{\braket{\phi}{\psi}}^2\leq 2^{-t}
\end{align*}
permitting a $\pm \epsilon$ precision estimator for $p$ using $\order{\xi \epsilon^{-2}}$ samples.
By a tighter bound for the RHS of \cref{eq:csim_num_samples_UB} the number of required samples and the simulation runtime can be improved by a factor of up to $\twonorm{\ket{\psi}}^2$.

\section{Low-rank matrix completion}\label{app:LR_Matrix_Completion}
\Cref{lem:NSPS_to_nullity} gives us the following interpretation of the $\probdec{220}$ problem over graph states: given a graph $G$, find an induced subgraph with maximal nullity. If we were to instead consider an arbitrary stabilizer input, then the tableau has the form $[C | D]$ for matrices satisfying the symplectic condition ($CD^T + DC^T = 0$). Of course, the setting is strictly more general (and therefore $\NP$-hard), but it is perhaps simpler to see that the $\probdec{220}$ problem seeks the lowest-rank matrix we can obtain from taking, for each qubit $i$, the $i$th column of either $C$ or $D$. In particular, an $X$-measurement on qubit $i$ corresponds to taking the $i$th column of $D$, and a $Z$-measurement on qubit $i$ corresponds to taking the $i$th column of $C$.

As it turns out, this fits into a large body of research surrounding rank minimization. In general, the input for such problems is a matrix filled with entries from some finite field (in our case, $\mathbb F_2$) as well as polynomially-many variables. The goal is then to find an assignment of the variables that minimizes the rank. Our $\NP$-hardness proof can be thought of as an alternative to other $\NP$-hardness results in this framework (most closely, the ``MinRank'' problem of \cite{BUSS1999Complexity}) with some additional nice properties: there is exactly one variable per column, and each variable cannot appear in more than one column.

It is also worth mentioning a broad area of low-rank matrix completion problems initiated by Haemers that was inspired by the Shannon capacity of graphs \cite{haemers1981upper}. This is a special case of the MinRank problem where the diagonal is fixed to 1, the off-diagonal entries at edges of an input graph are fixed to 0, and each remaining off-diagonal entry is assigned its own variable. It was first shown by Peeters \cite{peeters1996orthogonal} that this problem is $\NP$-hard (even over $\mathbb F_2$). There were many subsequent follow-up works showing hardness for different fields, different patterns of variables, approximations, and positive-semidefinite completions \cite{e2013complexity, hardt2014computational, shitov2016hard, chawin2023improved, chawin_et_al}. To our knowledge, there is no direct reduction from these problems to our own (or vice versa).

\section{Measures of entanglement}\label{app:Entanglement_and_mult_CSim}
Entanglement has proven to be an important yet elusive property of quantum systems. Substantial effort has been  expended in characterizing, quantifying and attributing operational significance to measures of entanglement.
Here we will discuss the \emph{relative entropy} $E_R (\rho)$, the \emph{geometric measure} $E_g (\ket{\psi})$ and the \emph{Schmidt measure} $E_{s} (\ket{\psi})$ of entanglement.
Evaluating these entanglement measures for general states has proven to be  computationally prohibitive. However, the rich multi-partite entanglement properties exhibited by stabilizer states has provided some fertile ground for this research.

Restricting to stabilizer states somewhat simplifies the landscape revealing several elegant and powerful results. For all pure stabilizer states, the relative entropy of entanglement and the geometric measure are equal \cite{hayashi2008entanglement, markham2007entanglement}.  For all graph states $\ket{G}$, Ref.~\cite{markham2007entanglement} showed the following connection between graph properties and entanglement:
\begin{align}\label{eq:geometric_measure_sandwich_ineq}
	\tau(G)\geq E_R(\ket{G})=E_g(\ket{G}).
\end{align}
The upper bound comes from the fact that measurement, in the Pauli-$Z$ basis, of the qubits associated with the minimum vertex cover of $G$ results in at most $\tau(G)$ non-deterministic outcomes contributing at most $\tau(G)$ factors of $1/2$ to the fidelity. Thereafter, the remaining qubits are a maximum independent set as all edges of $G$ have been removed by the previous measurement. Hence these qubits are in a Pauli-$X$ basis state $\set{\ket{+},\ket{-}}^{\otimes \alpha(G)}$ and can be measured in the Pauli-$X$ basis with deterministic outcome.
This inequalities can be tightened by searching over LC orbits of $G$.

The canonical polyadic (CP) rank\footnote{Also known as the \emph{tensor rank}.}, $R_{\min}$, of an $n$-qubit state $\ket{\psi}$ interpreted as an $n$-dimensional tensor quantifies the minimum number of terms needed to decompose $\ket{\psi}$ into a linear combination of product states. The Schmidt measure $E_s (\ket{\psi}):=\log_2 R_{\min}$ is a monotone under stochastic local operations and classical communication (SLOCC) \cite{eisert2001schmidt, Dur2000three_qubit_entanglement}.
The Schmidt measure for graphs states can also be upper bounded by the vertex cover number \cite{hein2004multiparty, hein2006entanglementgraphstatesapplications, markham2007entanglement}
\begin{align}\label{eq:Schmidt_measure_sandwich_ineq}
	\tau(G)\geq E_s(\ket{G}).
\end{align}

Our work associated with $\probnum{222}$ (and to a lesser extent $\probnum{220}$) is a parallel line of inquiry that fits into and complements the prior work. It can be seen as a further simplification of the entanglement measurement problem where we restrict the optimizations associated with the geometric and Schmidt measure from all product states to the discrete set of all stabilizer product states. Concretely, we can define and denote the \emph{stabilizer} product state equivalents of $E_g$ and $E_s$ by $E_{sg}$ and $E_{ss}$ respectively where for a stabilizer state $\ket{\psi}$, $E_{sg} (\ket{\psi}):= -\log_2 m({\psi},\pset{222})$ and $E_{ss}(\ket{\psi})$ is defined as the base-2 logarithm of the smallest number of terms required to decompose $\ket{\psi}$ into a linear combination of stabilizer product states.

We observe that \cref{eq:geometric_measure_sandwich_ineq} can be modified as follows:
\begin{align}
	\tau(G)\geq E_{sg}(\ket{G}) \geq E_g(\ket{G}).
\end{align}
The upper bound to $E_{sg}(\ket{G})$ is immediate from the proof of the upper bound to $E_g(\ket{G})$ in \cref{eq:geometric_measure_sandwich_ineq} since local Pauli measurements are used to derive the upper bound.
The lower bound can be observed by noting that $E_g$ is a relaxation of $E_{sg}$. Additionally, \cref{eq:Schmidt_measure_sandwich_ineq} can be modified as follows:
\begin{align}
	E_{sg}(\ket{G}) \geq E_{ss}(\ket{G}) \geq E_s(\ket{G}).
\end{align}
The lower bound can be observed by noting that $E_s$ is a relaxation of $E_{ss}$. The upper bound can be observed by using \cref{lem:orthogonal_PSS_decomp} stated and proven below. This establishes the relation $E_{sg}(\ket{G}) \geq E_{ss}(\ket{G})$ between the stabilizer product state counterparts that does not hold generally between the geometric and Schmidt measure and, to our knowledge, is unknown to hold even when inputs are restricted to graph states.

To prove the $\NP$-hardness of $\probdec{220}$ and $\probdec{222}$, we use a specific graph construction that maps an input graph $G$ into a new graph $H=H(G)$ (also known as the 3-subdivision of $G$ \cite{ghanbari2023securedominationnumberksubdivision}) where each edge in $G=(V,E)$ is divided into 3 edges by placing two new vertices on the edge. Using this construction, we show that $E_{sg}(\ket{H(G)})=\tau(G)+\abs{E}=\tau(H(G))$.
For graphs of the form $H(G)$ that also have easy to compute vertex cover numbers, this provides a means to computing $E_{sg}(\ket{H(G)})$. As a simple example, let us consider the $n$-vertex ring graph $G=R_n$ where $H(R_n)=R_{3n}$. Since $\tau(R_n)=\lceil n/2 \rceil$, our results immediately give $E_{sg}(\ket{R_{3n}})=\lceil n/2 \rceil+n$.
Even the inefficient computation of $\tau(G)$ requires vastly less computational resources \cite{Xiao_2017} compared to the, admittedly more general\footnote{These methods take as an input a full inefficient description of an $n$-qubit state making them applicable to general states rather than graph states, albeit at the expense of requiring not only exponential runtime but also exponential memory.}, modern numerical methods for bounding $E_g$ \cite{weinbrenner2026completehierarchiesgeometricmeasure}. Beyond the special case where the target graph state is of the form $\ket{H(G)}$, one may still employ the above results by either identifying an LC-equivalence between the target graph and one of the form $H(G)$ or alternatively resorting to bounds that can be computed by observing that the target state can, with a limited number of stabilizer measurements and CZ gates be transformed into the desirable form $\ket{H(G)}$.
From this perspective, our results provide a computational tool for computing and/or bounding $m({\psi},\pset{220})$, $m({\psi},\pset{222})$ and $E_g (\ket{\psi})$ for certain stabilizer states.

We now show that a solution to a $\ket{\psi}$ instance of  $\probproj{222}$ immediately provides an orthogonal stabilizer product state decomposition of $\ket{\psi}$ with exactly $1/m({\psi}, \pset{222})$ terms. This is a direct corollary of the next lemma.

\begin{lemma}\label{lem:orthogonal_PSS_decomp}
	Let $\ket{\psi}$ be an $n$-qubit stabilizer state with known stabilizers. Let $\ket{\phi}$ be an $n$-qubit stabilizer product state with $\alpha=\braket{\phi}{\psi}\neq 0$. Then, there is an orthogonal product state decomposition of $\ket{\psi}$:
	\begin{align}\label{eq:SPS_decomp}
		\ket{\psi}=\alpha \sum_{u \in \set{0,1}^{m}} (-1)^{u\cdot b} i^{u Q u^{T}} L \ket{u C},
	\end{align}
	where, $m=-\log_2 \abs{\alpha}^2$ and $b\in \bitstring{m}$, $Q \in \set{0,1}^{m \times m}$ and $C \in \set{0,1}^{m \times n}$ are all efficiently computable. Additionally,  $L$ is a local Clifford unitary that satisfies $L \ket{0^n}=\ket{\phi}$.
\end{lemma}

The proof of this statement relies on using $L^{\dagger}$ to locally rotate $\ket{\psi}$ into a preferred basis then observing basic properties of the standard quadratic form expansion \cite{dehaene2003clifford, nest2008classical, bermejo2012classical, bravyi2016improved, Beaudrap_2022} of this state.

\begin{proof}
	Let $\ket{\phi'}$ be a stabilizer product state that achieves the fidelity $F=\alpha \alpha^*=\abs{\braket{\phi'}{\psi'}}^2$ with the target state $\ket{\psi'}$.
	W.L.O.G., let us use local Clifford rotations to simultaneously rotate $\ket{\phi'} \rightarrow \ket{\phi}$ and $\ket{\psi'} \rightarrow \ket{\psi}$ such that $\ket{\phi}$ becomes the computational basis state $\ket{\phi}=\ket{0^n}$.
	An orthogonal stabilizer product state decomposition of $\ket{\psi}$ with exactly $R$ terms can then be converted back to an orthogonal stabilizer product state decomposition of $\ket{\psi'}$ with exactly $R$ terms by reversing the local Clifford rotations.
	Additionally, we observe that $F=\abs{\braket{\phi'}{\psi'}}^2=\abs{\braket{0^n}{\psi}}^2$ so we aim to construct an $R$-term decomposition where $1/R=\abs{\braket{0^n}{\psi}}^2=F$ to prove the claim.

	Let us denote the $n$ stabilizer generators of $\ket{\psi}$ by $\set{g_j}_{j\in [n]}$. We write these generators using the notation $g_j=i^{2b_j+x(j)\cdot z(j)} X^{x(j)} Z^{z(j)}$ where for each $j$, $b_j\in \set{0,1}$ determines the sign of the generator, $x(j), z(j) \in \set{0,1}^n$ are row vectors and for a local Pauli $P$ and column vector $v\in \set{0,1}^n$, the notation $P^v$ represents the operator $P^{v_1}\otimes \ldots \otimes P^{v_n}$ with $P^0$ interpreted as the identity. Let $C_x$ and $C_z$ be the $n\times n$ matrices with $x(j)$ and $z(j)$ as their $j^{\mathrm{th}}$ row vectors respectively and let $C=[C_x|C_z]$ be the associated $n \times 2n$ matrix. Define $Q:=C_x C_z^T$ where the diagonal entries of $Q$ are computed modulo $4$ while the off-diagonal entries are computed over $\mathbb{F}_2$.
	The symmetry condition $Q=Q^T$ is imposed by the fact that the generators commute.  The stabilizers of $\ket{\psi}$ take the form:
	\begin{align}
		s(y) & =g_1^{y_1}\times \ldots \times g_n^{y_n}                \\
		     & =(-1)^{y\cdot b} \times i^{y Q y^T} X^{y C_x} Z^{y C_z}
	\end{align}

	for $y \in \bitstring{n}$ a row vector. Here, we compute the exponent $y Q y^T$ over the integers. Without loss of generality, we can assume that $C$ is in  reduced row echelon form and $m$ is the matrix rank of $C_x$. This implies that $[Q]_{i,j}=0$ whenever $\max \set{i,j}>m$. We now split $y=(u,v)\in \bitstring{m} \times \bitstring{n-m}$ and use $C'_x$ and $C'_z$ to denote submatrices made up of the top $m$ rows of $C_x$ and $C_z$ respectively. We use $Q'=C'_x (C'_z)^T$ to denote the leading principal submatrix of order $m$ of $Q$. Hence:
	\begin{align}
		s(u,v) & =g_1^{u_1}\times \ldots \times g_m^{u_m} \times g_{m+1}^{v_1}\times \ldots \times g_n^{v_{n-m}} \\
		       & =(-1)^{(u,v)\cdot b} \times i^{u Q' u^T} X^{y C_x} Z^{y C_z}                                    \\
		       & =(-1)^{(u,v)\cdot b} \times i^{u Q' u^T} X^{u C'_x} Z^{y C_z}.
	\end{align}
	That is, the quadratic form component of the phases and the Pauli-$X$ components of the operator do not depend on $v$.

	Since the overlap $\braket{0^n}{\psi}$ cannot be zero by assumption, we have:
	\begin{align*}
		\ket{\psi} & = \frac{1}{\alpha^{*}} \ketbra{\psi}{\psi} \ket{0^n}                                                                                                                             \\
		           & = \frac{1}{\alpha^{*}} 2^{-n} \sum_{y \in \bitstring{n}} s(y) \ket{0^n}                                                                                                          \\
		           & = \frac{1}{\alpha^{*}} 2^{-n}\left[ \sum_{v \in \bitstring{n-m}} (-1)^{(0^m,v)\cdot b} \right] \sum_{u \in \bitstring{m}} (-1)^{(u,0^{n-m})\cdot b} ~i^{u Q' u^T}  \ket{u C_x'}.
	\end{align*}
	The orthogonality of $\set{\ket{uC_x'}}_u$ is ensured since $C_x'$ has full rank.
	Taking the inner product with $\ket{0^n}$ gives:
	\begin{align}
		\alpha & =\frac{1}{\alpha^{*}} 2^{-n} \left[ \sum_{v \in \bitstring{n-m}} (-1)^{(0^m,v)\cdot b} \right].
	\end{align}
	Hence,
	\begin{align*}
		\ket{\psi} = \alpha \sum_{u \in \bitstring{m}} (-1)^{(u,0^{n-m})\cdot b} ~i^{u Q' u^T}  \ket{u C_x'}.
	\end{align*}
	The normalization $\braket{\psi}{\psi}=1$ and the orthogonality of $\set{\ket{uC_x'}}_u$ implies that $m=-\log_2 \abs{\alpha}^2$. Finally, we apply the local Clifford $L$ to rotate to the original state:
	\begin{align*}
		\ket{\psi'} & = L\ket{\psi}                                                                      \\
		            & =\alpha \sum_{u \in \bitstring{m}} (-1)^{u\cdot b'} ~i^{u Q' u^T}  L \ket{u C_x'},
	\end{align*}
	where we have used $b'$ to denote the first $m$ entries of the vector $b$. Relabeling the primed variables produces the claimed \cref{eq:SPS_decomp}.
\end{proof}

\section{Classification into equivalence classes}
\label{app:Classification}

We now classify all non-empty subsets $\sscase\subseteq \sall$ into orbits under conjugation by $U_{loc}\in \clif{1}$.
For $\sscase\subseteq \sall$, and $U_{loc}$ a single qubit Clifford, we define $U_{loc} \curvearrowright \sscase:=\set{U_{loc} \ketbra{\phi}{\phi}U_{loc}^{\dagger}~|~\ketbra{\phi}{\phi}\in \sscase}$. This group action defines an equivalence relation $\eqloc$ on subsets $\sscase$ such that for any $\sscase, \sscase' \subseteq \sall$, we have
$\sscase \eqloc \sscase'$ iff $\sscase'$ is contained in the orbit of $\sscase$:
\begin{align}
	\orb{\sscase}:=\set{U_{loc} \curvearrowright \sscase~|~U_{loc} \in \clif{1}}.
\end{align}
As we will shortly show, \cref{eq:S_representatives} provides the complete list of representatives for the orbits of all non-empty subsets of $\sall$.
As outlined in \cref{sec:problem_classification}, if two subsets $\sscase,\sscase'$ share an orbit, then they are related by the group action of some element $U_{loc}\in \clif{1}$. As per \cref{eq:symmetry_optimization_fn}, this implies the existence of a symmetry $U=U_{loc}^{\otimes n}$ in the optimization function which establishes a poly time reduction between $\probdec{\stab{n},\sscase}$ and $\probdec{\stab{n},\sscase'}$. Hence, we will show that for any non-empty $\sscase'\subseteq \sall$, the complexity of $\probdec{\stab{n}, \sscase'}$ is (polynomial-time) equivalent to the complexity of $\probdec{\stab{n}, \sscase}$ for a unique $\sscase\in \reps$.

\begin{figure}[h]
	\centering
	\begin{tikzpicture}[
			state/.style={inner sep=4pt, font=\large}
		]

		\node[state] (0) at (0, 0) {$\ket{0}$};
		\node[state] (+) at (2, 0) {$\ket{+}$};
		\node[state] (-) at (5, 0) {$\ket{-}$};
		\node[state] (1) at (7, 0) {$\ket{1}$};

		\node[state] (i)  at (3.5,  1.5) {$\ket{+i}$};
		\node[state] (-i) at (3.5, -1.5) {$\ket{-i}$};

		\draw[<->, thick] (0) -- node[above] {$H$} (+);
		\draw[<->, thick] (-) -- node[above] {$H$} (1);
		\draw[<->, thick] (i) -- node[right] {$H$} (-i);

		\draw[->, thick] (+) -- node[above left] {${\phasegate}$} (i);
		\draw[->, thick] (i) -- node[above right] {${\phasegate}$} (-);
		\draw[->, thick] (-) -- node[below right] {${\phasegate}$} (-i);
		\draw[->, thick] (-i) -- node[below left] {${\phasegate}$} (+);

		\draw[->, thick] (0) to [out=135, in=225, looseness=5] node[left] {${\phasegate}$} (0);
		\draw[->, thick] (1) to [out=45, in=315, looseness=5] node[right] {${\phasegate}$} (1);

	\end{tikzpicture}
	\caption{The six single-qubit stabilizer states. The arrows indicate the transitions under the Hadamard ($H$) and phase (${\phasegate}$) gates modulo a global phase.}
	\label{fig:stabilizer_update_rules}
\end{figure}
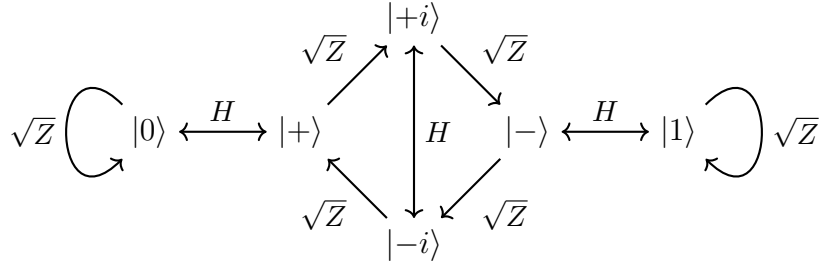

\begin{lemma}
	The set of all non-empty subsets of $\sall$ can be partitioned into nine equivalence classes with representatives $\reps$.
\end{lemma}

\begin{proof}
	The total number of non-empty subsets $\sscase\subseteq \sall$ is $2^6-1=63$. These can be partitioned into $63=\binom{6}{1} +\ldots+\binom{6}{6}$ by cardinality, since equivalent subsets necessarily have the same cardinality.  We claim that:
	\begin{enumerate}
		\item all subsets of cardinality 1 are equivalent to $\pset{100}$
		\item all subsets of cardinality 2 are equivalent to either $\pset{200}$ or $\pset{110}$
		\item all subsets of cardinality 3 are equivalent to either $\pset{210}$ or $\pset{111}$
		\item all subsets of cardinality 4 are equivalent to either $\pset{220}$ or $\pset{211}$
		\item all subsets of cardinality 5 are equivalent to $\pset{221}$
		\item all subsets of cardinality 6 are equivalent to $\pset{222}$.
	\end{enumerate}

	Clearly, cardinality is preserved\footnote{Subsets of different cardinality must be in separate orbits since $U_{loc}$ is invertible.} under the defined group action, hence each of the cardinalities is associated with at least one orbit. We will make use of another invariance that generalizes this observation.
	We note that the set of pairwise Hilbert--Schmidt inner products in any subset $\sscase\neq \emptyset$ is invariant under the defined group action. For clarity, let us define the function:
	\begin{align}
		N:\set{0,1/2,1} \times \mathcal{P}(\sall) & \rightarrow \bbn                                                              \\
		(v,\sscase)                               & \mapsto N_v(\sscase):=\abs{\set{(a,b)\in \sscase\times \sscase~|~\tr{ab}=v}},
	\end{align}
	where, $\mathcal{P}(\sall)$ denotes the power set of $\sall$.
	Since equivalent subsets are related by pairwise conjugation by a fixed unitary, if $\sscase \eqloc \sscase'$ then for all $v\in \set{0,1/2,1}$, we must have $N_v(\sscase)=N_v(\sscase')$. We note that $N_1(\sscase)=\abs{\sscase}$ which gives us the invariance of cardinality. Shortly we will also use $N_0(\cdot)$ to distinguish orbits with the same cardinality.

	We now prove claims 6, 1, 2 and 3 in that order. Claim 6 is trivial as $\pset{222}$ is the only subset of size 6.
	Using \cref{fig:stabilizer_update_rules}, one can easily confirm that any element of $\sall$ can be mapped to any other by a sequence of conjugations by Hadamard and phase gates. Hence every $\sscase\subseteq \sall$ with $\abs{\sscase}=1$ is in  $\orb{\pset{100}}$.

	For each of claims 2 and 3, we observe that there are (at least) two distinct orbits. Since:
	\begin{itemize}
		\item $N_0(\pset{200})=2\neq 0=N_0(\pset{110})$, i.e., subsets of cardinality 2 have at least 2 distinct orbits,
		\item $N_0(\pset{210})=2\neq 0=N_0(\pset{111})$, i.e., subsets of cardinality 3 have at least 2 distinct orbits.
	\end{itemize}

	We deal with each of the four orbits (associated with $\pset{200},\pset{110},\pset{210},\pset{111}$) separately showing that the union of orbits $\pset{200},\pset{110}$ encompass all subsets of cardinality 2 and the union of orbits $\pset{210},\pset{111}$ encompass all subsets of cardinality 3.

	The two element set $\pset{200}$ contains an orthogonal pair of projectors. It is easy to verify that there are 3 such choices; these are associated with the $\pm 1$ eigenprojections for the 3 Paulis $X, Y$ and $Z$. We represent each of these subsets with the notation $(\pm,0,0), (0,\pm,0)$ and $(0,0,\pm)$ respectively. Since ${\phasegate} \curvearrowright (\pm,0,0)=(0,\pm,0)$ and $H\curvearrowright (\pm,0,0)=(0,0,\pm)$, these 3 sets are connected by a group action and contain $\pset{200}$. Hence they are equivalent and represented by $\pset{200}$.

	The two element set $\pset{110}$ does not contain an orthogonal pair of projectors. We now claim that there are $3\times 2\times 2=12$ such choices associated with choosing 2 of 3 Paulis from $\set{X, Y, Z}$ and for each of the chosen Paulis, choosing one of two eigenprojections. We represent these in \cref{tbl:orbits110}.
	To show that all of these are connected under the group action of specific elements from $\clif{1}$, we show that the three elements in the top row of \cref{tbl:orbits110} are connected and that for any column, the top row element is connected to all other row elements. Observe that $ZH$ maps $(+,+,0)$ to $(0,+,+)$ and ${\phasegate}^{\dagger}$ maps $(0,+,+)$ to $(+,0,+)$. Additionally, the three Paulis each move the eigenprojection choice ``+1 and +1'' to each of the other choices in the same column. Hence, all 12 elements are contained in $\orb{\pset{110}}$. This accounts for all size 2 subsets as there are $\binom{6}{2}=15=3+12$.

	\begin{table}[H]
		\centering
		\caption{Table associated with subsets in $\orb{\pset{110}}$}\label{tbl:orbits110}
		\begin{tabular}{|c|c|c|c|}
			\hline
			\textbf{Eigenprojection choices} & \textbf{Pauli X and Y} & \textbf{Pauli X and Z} & \textbf{Pauli Y and Z} \\ \hline
			$+1$ and $+1$                    & $(+,+,0)$              & $(+,0,+)$              & $(0,+,+)$              \\ \hline
			$+1$ and $-1$                    & $(+,-,0)$              & $(+,0,-)$              & $(0,+,-)$              \\ \hline
			$-1$ and $+1$                    & $(-,+,0)$              & $(-,0,+)$              & $(0,-,+)$              \\ \hline
			$-1$ and $-1$                    & $(-,-,0)$              & $(-,0,-)$              & $(0,-,-)$              \\ \hline
		\end{tabular}
	\end{table}

	The three element set $\pset{210}$ contains an orthogonal pair of projectors. We now claim that there are $3\times 2 \times 2=12$ such choices associated with choosing 1 of 3 Paulis from $\set{X, Y, Z}$ that will contribute both eigenprojectors, then choosing 1 of the remaining 2 Paulis and choosing one of its two eigenprojections. We represent each of these subsets in \cref{tbl:orbits210}.
	To show that all of these are connected under the group action of specific elements from $\clif{1}$, we first consider how the phase gate connects these subsets. Specifically, we consider the ${\phasegate}$-orbits of elements in \cref{tbl:orbits210}:
	\begin{align}
		\orb{\sscase}_{\phasegate}:= \set{{\phasegate}^p \curvearrowright \sscase~|~p\in \bbz_4}.
	\end{align}
	Since the phase gate acts trivially on the Pauli $Z$ eigenprojectors, the number and sign of $Z$ eigenprojectors remains invariant. The phase gate acts as a 4-cycle permuting the entries in the rightmost  column of \cref{tbl:orbits210}; let us call these the $\pm Z$ ${\phasegate}$-orbit. Another 4-cycle ${\phasegate}$-orbit is given by $\set{(\pm,+,0), (-,\pm,0), (\pm,-,0), (+,\pm,0)}$ which we call the empty-Z ${\phasegate}$-orbit. Additionally, two 2-cycle ${\phasegate}$-orbits are given by $\set{(\pm,0,+), (0,\pm,+)}$ and $\set{(\pm,0,-), (0,\pm,-)}$ which we call the $+Z$ and the $-Z$ ${\phasegate}$-orbits respectively. We now note that the action of the Hadamard gate connects each of the $+Z$, $-Z$ and empty-$Z$ ${\phasegate}$-orbits to the $\pm Z$ ${\phasegate}$-orbit. Specifically,
	$H \curvearrowright (\pm,0,+)=(+,0,\pm)$,
	$H \curvearrowright (\pm,0,-)=(-,0,\pm)$ and
	$H \curvearrowright (\pm,-,0)=(0,+,\pm)$. Hence all 12 subsets in \cref{tbl:orbits210} lie in $\orb{\pset{210}}$.

	\begin{table}[H]
		\centering
		\caption{Table associated with subsets in $\orb{\pset{210}}$. The columns are associated with the choice of Pauli $P\in \set{X, Y, Z}$ that contributes both eigenprojectors. The rows, are associated with the choice of the final projector. This is determined by choosing a Pauli $Q\in \set{X,Y,Z}\setminus \set{P}$ and the sign associated with its eigenprojector. We label $Q$ relative to $P$ using the cyclic Pauli order $XYZ$ so for example, if $P=X$ then the next Pauli is $Y$ and the previous Pauli is $Z$.}\label{tbl:orbits210}
		\begin{tabular}{|l|c|c|c|}
			\hline
			\textbf{eigenprojection choices} & \textbf{Both of Pauli X} & \textbf{Both of Pauli Y} & \textbf{Both of Pauli Z} \\ \hline
			Next Pauli and $+1$              & $(\pm,+,0)$              & $(0,\pm,+)$              & $(+,0,\pm)$              \\ \hline
			Next Pauli and $-1$              & $(\pm,-,0)$              & $(0,\pm,-)$              & $(-,0,\pm)$              \\ \hline
			Previous Pauli and $+1$          & $(\pm,0,+)$              & $(+,\pm,0)$              & $(0,+,\pm)$              \\ \hline
			Previous Pauli and $-1$          & $(\pm,0,-)$              & $(-,\pm,0)$              & $(0,-,\pm)$              \\ \hline
		\end{tabular}
	\end{table}

	The three element set $\pset{111}$ does not contain any orthogonal pair of projectors. We now claim that there are $2\times 2 \times 2=8$ such choices associated with choosing 1 of 2 eigenprojectors for each of the 3 Paulis $\set{X, Y, Z}$. Each of these can be represented as one of the strings $(+,+,+),(+,+,-),(+,-,+),(+,-,-),(-,+,+),(-,+,-),(-,-,+),(-,-,-)$. The action of phase gates breaks these up into two 4-cycles. Recycling our naming convention, we call these the $+Z$ and $-Z$ ${\phasegate}$-orbits. We observe that a Pauli $X$ connects these ${\phasegate}$-orbits. Hence all 8 subsets lie in $\orb{\pset{111}}$. This accounts for all size 3 subsets as there are $\binom{6}{3}=20=12+8$.

	Only claims 4 and 5 remain to be proven. We prove these by reducing them to the already proven claims 2 and 1 respectively. For this, we use the set complement function:
	\begin{align*}
		\notset{}: \mathcal{P}(\sall) & \rightarrow \mathcal{P}(\sall)                     \\
		\sscase                       & \mapsto \notset{\sscase}:=\sall \setminus \sscase.
	\end{align*}
	We claim that this bijection has the property that for all $\sscase, \sscase' \subseteq \sall$, $\sscase \eqloc \sscase' \iff \notset{\sscase} \eqloc \notset{\sscase'}$.
	This implies that for all $\sscase \subseteq \sall$, the set of entry-wise complements of $\orb{\sscase}$ form an orbit $\orb{\notset{\sscase}}$. In particular:
	\begin{enumerate}
		\item Since $\notset{\pset{221}}=(0,0,-)\in \orb{\pset{100}}$; this implies that $\orb{\pset{221}}$ is the entry-wise complement of $\orb{\pset{100}}$.
		\item Since $\notset{\pset{220}}=(0,\pm,0)\in \orb{\pset{200}}$; this implies that $\orb{\pset{220}}$ is the entry-wise complement of $\orb{\pset{200}}$.
		\item Since $\notset{\pset{211}}=(0,-,-)\in \orb{\pset{110}}$; this implies that $\orb{\pset{211}}$ is the entry-wise complement of $\orb{\pset{110}}$.
	\end{enumerate}

	It only remains to prove the claimed property of the complement map. Suppose $\sscase \eqloc \sscase'$, this means that there exists $U_{loc}\in \clif{1}$ such that $U_{loc} \curvearrowright \sscase=\sscase'$. Now, we observe that when $U_{loc}$ acts by conjugation on the elements of $\sall$,  it permutes them. Since $\sall$ is the disjoint union of $\sscase$ and $\notset{\sscase}$, this establishes that $U_{loc} \curvearrowright \sscase =\notset{(U_{loc} \curvearrowright \notset{\sscase})}$ as per the commutative diagram \cref{fig:commutative_diagram}. Hence, $U_{loc} \curvearrowright \sscase=\sscase' \implies U_{loc} \curvearrowright \notset{\sscase}=\notset{\sscase'}$. Applying the same argument in reverse using $U_{loc}^{\dagger}$ proves the reverse implication.

	\begin{figure}[h]
		\centering
		\begin{tikzpicture}[
				state/.style={inner sep=4pt, font=\large}
			]

			\node[state] (a) at (0, 0) {$\sscase$};
			\node[state] (b) at (2, 0) {$\notset{\sscase}$};
			\node[state] (a') at (0, -2) {$\sscase'$};
			\node[state] (b') at (2, -2) {$\notset{\sscase'}$};

			\draw[<->, thick] (a) -- node[above] {$\notset{}$} (b);
			\draw[<->, thick] (a') -- node[above] {$\notset{}$} (b');

			\draw[->, thick] (a) -- node[left] {$U_{loc}$} (a');
			\draw[->, thick] (b) -- node[right] {$U_{loc}$} (b');

		\end{tikzpicture}
		\caption{The commutative diagram for the complement map and entry-wise conjugation by any fixed $U_{loc}$. Since the top row and the bottom row are both disjoint unions of $\sall$ and entry-wise conjugation by any fixed $U_{loc}$ permutes the entries of $\sall$, it must hold that $\notset{(U_{loc} \curvearrowright \sscase)} ={U_{loc} \curvearrowright \notset{\sscase}}$.}
		\label{fig:commutative_diagram}
	\end{figure}
\end{proof}

\end{document}